\documentclass[12pt]{article}

\usepackage[english]{babel}
\usepackage{amsmath}
\usepackage{amssymb}
\usepackage{amsthm}
\usepackage{graphicx}

\usepackage[colorlinks=true,linkcolor=blue, citecolor=blue, urlcolor=magenta]{hyperref}
\usepackage{multirow}

\usepackage{caption}
\usepackage{subcaption}
\usepackage{caption}
\usepackage[dvipsnames]{xcolor}
\usepackage{soul} 
\usepackage{gensymb} 

\usepackage{natbib}
\bibpunct[, ]{(}{)}{,}{a}{}{,}%
\newtheorem{theorem}{Theorem}

\newtheorem{lemma}[theorem]{Lemma}
\newtheorem{proposition}[theorem]{Proposition}

\theoremstyle{definition}
\newtheorem{definition}[theorem]{Definition}

\theoremstyle{remark}

\newtheorem{example}[theorem]{Example}

\date{}
\begin{document}

	\title{{ On the fair division of a random object\thanks{Comments by seminar participants at the University of Liverpool, Zurich Technical University,  University of
				Lancaster, the Weizmann Institute, Tel-Aviv University, Hebrew University of Jerusalem,  Technion, 
				Paris School of Economics,  Maison des Sciences \'{E}conomiques,  University of Rochester, and Higher School of Economics in St.~Petersburg are gratefully acknowledged.
				Remarks by Yossi Azar and William Thomson  were especially helpful. 	
				The project benefited from numerical simulations
				by Yekaterina Rzhewskaya, a PhD student at the HSE St.~Petersburg. 
				We are grateful to Michael Borns for proofreading the paper.
				All the three authors acknowledge  support from the Basic Research Program of the National Research University Higher School of Economics. Sandomirskiy was partially supported by  Grant 19-01-00762 of the Russian
				Foundation for Basic Research and the European Research Council (ERC) under the European Union’s Horizon
				2020 research and innovation program (\#740435).}}}
\author{{Anna Bogomolnaia}\footnote{University of
		Glasgow, UK}\,\,$^{,\,\ddagger}$ \and  {Herv%
		\'{e} Moulin}$^\dagger,\,$\footnote{Higher School of Economics, St
		Petersburg, Russia} \and  {Fedor Sandomirskiy}\footnote{Technion, Haifa, Israel}\,\,$^{,\,\ddagger}$}

\maketitle

\begin{abstract}
Ann likes oranges much more than apples; Bob likes apples much more than oranges. Tomorrow they will receive one fruit that will be an orange or an
	apple with equal probability. Giving one half to each agent is fair for each realization of the fruit. However, agreeing that whatever fruit appears
	will go to the agent who likes it more gives a higher expected utility to
	each agent and is fair in the average sense: in expectation, each agent
	prefers his allocation to the equal division of the fruit, i.e., he gets a
	fair share.
	
We turn this familiar observation into an economic design problem:
	upon drawing a random object (the fruit), we learn the realized utility of
	each {agent and can compare it to the mean of his distribution of utilities; no other statistical information about the distribution is available.} We fully characterize the division rules using
	only this sparse information in the most efficient possible way, while
	giving everyone a fair share. Although the probability distribution of
	individual utilities is arbitrary and mostly unknown to the manager, these
	rules perform in the same range as the best rule when the manager  has full access to
	 this distribution.
\end{abstract}

\maketitle

 {\section{Introduction}}\label{sec1}

The trade-off between fairness and efficiency is a popular concern
throughout the social sciences (e.g., \citealt{Okun1975}), but its formal
evaluation is a fairly recent concern \citep{Caragianis2009, Bertsimas2011, Bertsimas2012}.

In the case of rules to divide fairly a random object, this trade-off
depends on the information available to the rule. We characterize here a
family of division rules that are fair in expectation, use minimal
information about the underlying distribution of utilities, and are the most
efficient with these two properties. {Efficiency is measured by the sum of 
utilities calibrated by their
mean values.}
 We also deliver
a surprisingly optimistic message to the risk-averse manager, who evaluates
the rules by their worst-case behavior: our rules have almost the same
worst-case efficiency as optimal rules when the manager has full access to the
distribution.

Before discussing our model and results formally, we illustrate them in a
stylized example.

\begin{example}
	\label{ex_1} Two agents, $a$ located in town $A$ and $b$ located in $B$,
	work as repairmen for the same company. The manager distributing incoming
	orders (jobs) looks for a fair and efficient procedure to allocate tasks
	between $a$ and $b$. The agents' salary is independent of the number of jobs
	they perform, and so each agent wants to have as little work as possible. {Hence, in this story, the manager must allocate a bad.}\footnote{{If instead each new job is desirable for both agents
		(as in piecemeal work), the manager must allocate a random good, which we
		briefly discuss afterward.}}
	
	The
	jobs may come from one of three towns $A$, $B$, or $C$, and each agent
	prefers to work in his own town. When a new order arrives, the manager
	learns the disutility of both agents for this particular job. These
	disutilities are presented in Table~\ref{tab_1}. 	$A$ and $B$ are small towns and are each responsible for $\frac{1}{4}$ of all
	orders, while $C$ is big and half of the orders come from there. 
	\begin{table}[h!]
	\begin{center}
		\caption{Disutilities and probabilities\label{tab_1}}	
		{\begin{tabular}{c|ccc}
				town & $A$  & $B$  & $C$\\ 
				\hline  
				agent $a$ & $1$ & $5$ & $5$ \\ 
				agent $b$ & $5$ & $3$ & $4$ \\ 
				\hline 
				probability & ${1}/{4}$ & ${1}/{4}$ & ${1}/{2}$ 
		\end{tabular}}
	\end{center}			
\end{table}

	If the only objective of the manager is to minimize the social cost (the sum
	of expected disutilities) and fairness is not an issue, then she allocates
	each job to the lowest disutility agent, following the familiar Utilitarian
	rule. See Table~\ref{tab_UT}. 
	\begin{table}[h!]
	\begin{center}
		\caption{Utilitarian rule\label{tab_UT}}		
		{\begin{tabular}{c|ccc|c|c}
				town & $A$  & $B$  & $C$ & expected costs &  social cost\\ 
				\hline  
				agent $a$ & $1$ & $0$ & $0$  & $0.25$& \multirow{2}{0em}{$3$} \\ 
				agent $b$ & $0$ & $1$ & $1$ &  $2.75$& \\ 
				\hline 
			\end{tabular} \
		}
	\end{center}
\end{table}

	So agent $b$ takes all jobs from towns $B$ and $C$ and incurs expected costs
	of $0\cdot \frac{5}{4}+1\cdot \frac{3}{4}+1\cdot\frac{4}{2}=\frac{11}{4}=2.75$. This exceeds his disutility of $2$ in the benchmark
	Equal Split allocation where the manager flips a fair coin to allocate each
	job. In this sense the Utilitarian rule is unfair to agent $b$.
	
	The Fair Share requirement says that each agent must (weakly) prefer his
	allocation to the Equal Split. In our example it caps the expected disutility of
	each agent at 2, to ensure that he is treated fairly in expectation, i.e.,
	ex~ante. Ex~ante fairness is especially compelling if the allocation
	decision is repetitive as in our example. It is rather permissive and leaves
	room for efficiency gains by exploiting differences in individual preferences.
	If the manager knows the prior distribution over the incoming jobs, i.e., the
	whole of Table~\ref{tab_1} including the probabilities, she finds the allocation
	minimizing the social cost under the Fair Share requirement by solving a linear
	program. This Optimal fair prior-dependent rule reallocates $\frac{3}{8}$ of
	the orders from town~$C$ to~$a$ to guarantee his fair share to $b$
	(Table~\ref{tab_OPT}). 
	\begin{table}[h!]
	\begin{center}
		\caption{Optimal fair prior-dependent rule\label{tab_OPT}}		
		{\begin{tabular}{c|ccc|c|c}
				town & $A$  & $B$  & $C$ & expected costs & social cost\\ 
				\hline  
				agent $a$ & $1$ & $0$ & ${3}/{8}$  & $19/16=1.1875$ & \multirow{2}{2em}{$3.1875$}\\ 
				agent $b$ & $0$ & $1$ & ${5}/{8}$ &  $2$& \\ 
				\hline 
			\end{tabular} 
		}
	\end{center}
\end{table}
	\smallskip
	
	How well can the manager do if, upon arrival of a new order, she only learns the vector of disutilities and has no clue  about the underlying
	probabilities of other possible orders? If she has no additional information
	at all, then the Equal Split is the only available fair rule. Indeed, without a
	common scale or a reference point for the disutility of each agent, how can
	she react to the observation that, for a particular job, the disutility of
	agent $a$ is $5$ and that of $b$ is $3$? Giving to agent $b$ more than half
	of this job (in probability) may violate Fair Share for $b$ if $3$ is greater
	than $b$'s average disutility for a job; similarly she cannot give more than
	half to agent $a$: what if $5$ is greater than $a$'s expected disutility?
    In other
	words, there are no non-trivial prior-independent fair rules.
	
	We assume that the manager can scale {disutilities:} upon
	realization of an object she knows each agent's  {disutility normalized by its mean value.} 
	 In our example she may observe the realized {absolute cost} of each
	job to each agent, and know as well their expected costs, a simple {first moment}
	estimated from previous draws. Or the agents themselves may report,
	directly and truthfully, the ratios of absolute to average costs.

	We focus on division rules taking as inputs the vector of {normalized
	disutilities,} and call these rules {\emph{almost prior-independent}} (API).
	To use API rules, the manager does not need any statistical information about the underlying distribution except the average costs.  These rules are practical if the manager's decisions are based on a small sample of observations (say two dozen), enough to get a reasonable estimate of the mean but not enough to estimate the whole distribution. Moreover, in our example, the repairman may have a good understanding of the average time it takes to complete a certain task but may find reporting the distribution or even its second moment problematic.  Surprisingly, the minimal amount of information required by API rules is enough to implement Fair
		Share, while incurring a social cost  close
		to that of the Optimal fair prior-dependent rule; this makes the API family appealing even if the manager has some extra statistical information.

	A simple example of a fair  API rule is the Proportional rule: it divides
	each job between the agents in proportion to their inverse {normalized} costs.
	In our example both expected costs are equal to $4$ and so we can use absolute
	costs instead of {normalized} ones. The Proportional rule picks the allocation
	in Table~\ref{tab_Prop}. 
	\begin{table}[h!]
		\begin{center}	
			\caption{Proportional rule\label{tab_Prop}}	
			{\begin{tabular}{c|ccc|c|c}
					town & $A$  & $B$  & $C$ & expected costs& social cost\\ 
					\hline  
					agent $a$ & $5/6$ & $3/8$ & ${4}/{9}$  & $515/288\approx 1.788$ & \multirow{2}{4em}{$\approx3.576$}\\ 
					agent $b$ & $1/6$ & $5/8$ & ${5}/{9}$ &  $515/288\approx 1.788$ & \\ 
					\hline 
				\end{tabular} 
			}
		\end{center}
	\end{table}
	
	\smallskip Our main results characterize the most efficient fair  API
	division rules. For bads, it is a single rule that we call the Bottom-Heavy
	rule. It has smaller a social cost than any other fair API rule, especially the Proportional one. In our example, it selects the
	allocation in Table~\ref{tab_BH}.%
	\begin{table}[h!]
		\begin{center}
			\caption{Bottom Heavy rule\label{tab_BH}}		
			{\begin{tabular}{c|ccc|c|c}
					town & $A$  & $B$  & $C$ & expected costs & social cost\\ 
					\hline  
					agent $a$ & $1$ & $1/6$ & ${3}/{8}$  & $67/48\approx 1.396$ & \multirow{2}{4em}{$\approx3.271$}\\ 
					agent $b$ & $0$ & $5/6$ & ${5}/{8}$ &  $15/8=1.875$ & \\ 
					\hline 
				\end{tabular} 
			}
		\end{center}
	\end{table}
	The social cost of the Bottom-Heavy rule ($67/48+15/8\approx 3.271$) is only $102.6\%$ of
	the cost for the Optimal fair prior-dependent rule, which is equal to $\frac{%
		19}{16}+2=3.1875$. The allocation of the Proportional rule is only within $112\%$
	of this optimum.

	As we show, the social cost of the Bottom-Heavy rule is \emph{always} close
	to the optimal social cost in the two-agent case. In other words, the
	improvement in efficiency from collecting the detailed statistical data is
	negligible, and it is enough to know expectations to approximate the optimal
	social cost. If the number of agents is large, the Optimal fair prior-dependent rule may
	significantly outperform the Bottom-Heavy rule for some distributions, but
	the worst-case guarantees of both rules remain close to each other. So the
	Bottom-Heavy rule remains a good choice for a risk-averse manager even if
	the population of agents is large.\smallskip
	
	All the rules and results that we just described have their analogs for
	goods. Yet problems with goods and problems with bads are not equivalent.
	That is to say, the results for goods and bads are similar but not  mirror
	images of one another. In particular, the social cost of the Bottom-Heavy
	rule for bads is lower than that of any other API fair rule not only in
	expectation but also ex~post, i.e., given the realization of the vector of
	disutilities. But for goods we find a one-parameter family of Top-Heavy
	rules that are not dominated by any other rule ex~post.\footnote{%
		For other unexpected differences between the fair division of goods and 
		bads, see \cite{BMSY2017} and \cite{BMSY2018}.}
	
	The lack of equivalence between goods and bads stems from the fact that agents dividing bads prefer smaller shares regardless of disutilities while agents dividing goods want bigger shares. To illustrate this point,
	consider a natural attempt to make
	goods from bads, namely, by adding a large enough constant $p$ to all
	disutilities. In our example, assume that the agents are paid $p=16$ for a
	completed job, so that each job is attractive and the manager is now dividing a
	good. The allocation she proposed when the jobs were bads may no longer be
	fair when they are goods. This is the case for the allocation in Table~\ref%
	{tab_OPT} that no longer satisfies Fair Share for agent $a$: his expected
	utility equals $\frac{7}{16}p-\frac{19}{16}=6-\frac{3}{16}$ (he completes $%
	1\cdot \frac{1}{4}+0\cdot \frac{1}{4}+\frac{3}{8}\cdot\frac{1}{2}=\frac{7}{%
		16}$ of all the orders) while  Fair Share requires his utility to
	be at least $\frac{p-4}{2}=6$ (his expected utility from completing half of
	the orders).   As we see, our rules are not translation-invariant because we distribute unequal shares.\footnote{The difference between goods and bads disappears in the class of allocations, where each agent receives exactly the same sum of shares as in \cite{Hylland}. We do not impose such a restriction and allow for allocations with unequal shares, i.e., unequal total probabilities of receiving one of the objects.} Indeed,
	agent $a$ receives the job with probability $\frac{7}{16}$, while $b$'s
	probability is $\frac{9}{16}$, and therefore adding a constant $p$ creates
	imbalance in their allocation by increasing $a$'s utility by $\frac{7}{16}p$
	and $b$'s by $\frac{9}{16}p$. 
\end{example}

\subsection{Overview of results and organization of the paper}

We identify an object with a non-negative vector of dimension $n$, the
number of agents. An instance of our problem is a probability distribution
over such vectors, which we call the prior. The object is either a
unanimously desirable \textit{good} or a unanimously
undesirable \textit{bad}.

The input of a \textit{prior-dependent} division rule is the realized
(dis)utilities and the full description of the prior. By contrast,\ the
input of an  \emph{almost prior-independent} (API) rule is simply the vector of
{normalized} (dis)utilities (absolute over expected). Therefore, in addition to
realized absolute (dis)utilities, {the rule} only needs to know the expected
(dis)utilities from the prior (and not even that if agents report these ratios
themselves).

The fairness of a rule is captured by a simple lower (upper) bound on the
utilities (disutilities) it distributes. The rule satisfies \textit{Fair Share}
if each agent is guaranteed at least $\frac{1}{n}$-th of his expected
utility for the whole good, or at most $\frac{1}{n}$-th of his expected
disutility for the whole bad. We measure efficiency by the sum of {normalized}
utilities for a good {and of normalized} disutilities for a bad. More comments on
our definitions of fairness and efficiency are in Section~\ref%
{subsect_assumptions}.

As explained above, the \textit{Equal Split} is the only fair and fully prior-independent
rule. Our first result is that much more efficient rules are available in
the class of fair  API rules. The simplest example is the \textit{Proportional
rule} allocating a good in proportion to {normalized} utilities (and a bad in
proportion to inverse {normalized} disutilities); see Section~\ref{sec2}.

In Section~\ref{sec3}, we characterize the most efficient fair API rules for dividing
 a good. Optimality is used in the following strong sense: one API rule
dominates another if it generates at least as much social {welfare} for each
realization of the utilities. This relation is very demanding and therefore one would
expect that most pairs of fair API rules are incomparable, and that the set
of undominated rules must be large. This intuition is wrong. For two agents,
a single rule, the \textit{Top-Heavy} rule, dominates every other fair API
rule. For more than two agents there is a one-dimensional family of
undominated Top-Heavy rules (and so any fair API rule is dominated by at least
one rule in the family). We call these rules {Top-Heavy} because they favor
the agents with high {normalized} utility to the extent that the Fair
Share requirement allows it.

The parallel analysis for the division of a bad in Section~\ref{sec4} yields
sharper results. For any number of agents there is a single \textit{Bottom
	Heavy} rule dominating every other API fair rule. This rule
favors the {agents with low normalized} disutility  as much as possible under Fair
Share.

Sections~\ref{sec5} and~\ref{sec6} compare the efficiency of our Top-Heavy
and Bottom-Heavy rules to that of the most efficient prior-dependent rules.
We start with the worst-case analysis in Section~\ref{sec5}: the worst case
is with respect to all possible prior distributions of the vector of
(dis)utilities. We focus on two indices. The {\textit{Competitive Ratio} (CR)%
} of a rule $\varphi $ compares it to the Optimal fair prior-dependent rule.
For goods, CR is the worst-case ratio of the optimal social welfare to the social welfare generated by
 $\varphi $; as usual, CR is at least $1$ for any $\varphi $. For bads it
is the ratio of the social cost generated by $\varphi $ to the optimal social cost. CR
quantifies the efficiency loss implied by almost prior-independence. The 
\textit{Price of Fairness} (PoF) is defined similarly but now the rule is
compared to the rule maximizing the social welfare (or minimizing the social
cost in the case of a bad) without any fairness constraints. PoF quantifies the
efficiency loss due to the fairness requirement.

Remarkably, we show that for any fair API rule, CR and PoF are equal, and so it is
enough to describe the results for PoF. In Example~\ref{ex_1}, we saw that
the social cost of the Bottom-Heavy rule was close to optimal. This is not a
coincidence. For two agents, the PoF of the Top-Heavy rule for a good is $109\%$ and the PoF of the Bottom-Heavy rule for a bad is $112.5\%$; for the
Proportional rule, these numbers are $121\%$ for a good and $200\%$ for a
bad. Thus the Top-Heavy and Bottom-Heavy rules outperform the Proportional one,
especially for a bad. As the number $n$ of agents grows,\ the PoFs of (some
of) the Top-Heavy rules and the Bottom-Heavy rule grow as $\sqrt{n}/{2}$ and ${%
	n}/{4}$, respectively. Thus fairness becomes costly for API rules when then
number of agents is large. However, the PoF for the Optimal fair
prior-dependent rule has the same asymptotic behavior (in the case of a good,
this was shown by~\citealt{Caragianis2009}); i.e., our API
rules provide the same worst-case guarantees as the prior-dependent rules.

Section~\ref{sec6} complements the worst-case analysis by looking at the
efficiency of our rules for particular prior distributions. We focus on a
benchmark case, where individual (dis)utilities are statistically
independent and drawn from familiar distributions, i.e., uniform, exponential, and
so on. While the worst-case results of Section~\ref{sec5} show that fairness
becomes extremely costly for large $n$, in the setting of Section~\ref{sec6}
the Top-Heavy and Bottom-Heavy rules generate, independently of $n$, a
constant positive fraction of the optimal social welfare (of the social cost for a
bad) \textit{unconstrained by Fair Share}. For example, if the distribution
of individual utilities is uniform on $[0,1]$, then the unconstrained social
welfare can only be $132\%$ higher than that of the Top-Heavy rule even if
the number of agents is large; for the exponential distribution, we get $%
188\%$. This confirms the common wisdom that allocation rules behave
much better on average than in the worst case.

Section~\ref{sec7} discusses possible extensions of our model such as
stronger fairness requirements, asymmetric ownership rights, and dividing a
mixture of goods and bads. Section~\ref{sec8} concludes. 

Appendices~\ref{appA},~\ref{appB}, and~\ref{appC} contain many proofs.

\subsection{Modeling choices}\label{subsect_assumptions}

\paragraph{Fairness.}

The Fair Share requirement (aka proportional fairness) was {introduced by
\cite{Steinhaus48}} at the onset of the fair division literature:
each agent should weakly prefer his allocation to the Equal Split of the
resources. It is a noncontroversial and fairly weak requirement.
Two popular strengthenings of Fair Share, \emph{Envy Freeness} and \emph{Max-min}
fairness, can also be discussed in the context of our model.

To define Envy Freeness in the case of a good, we fix the probability
distribution on utility profiles, and require, for any two agents $i$ and $j$%
, that agent $i$'s expected utility from his share be no less than his
(agent $i$'s) expected utility from agent $j$'s share. This is a much
tighter restriction on API rules than Fair Share that severely reduces
their efficiency; see the brief discussion in Section~\ref{sec7}.

Max-min fairness looks for an allocation where the smallest of individual
utilities (calibrated so that interpersonal comparisons make sense) is
maximized; see \citet{Ghodsi2011} and \citet{Bertsimas2011}. Our API rules
are not well suited to maximizing the smallest {normalized} utility (or minimizing
the largest {normalized} disutility).

\paragraph{{Normalization} of utilities.}

Our definition of social welfare and social cost uses {(dis)utilities normalized by mean values.} This  allows
interpersonal comparisons of utilities, such as the following: this object is worth $40\%$
more than average to Ann, but only $20\%$ more to Bob. {Normalization of (dis)utilities}
is a familiar tool of normative economics {and goes back to the concept of
\emph{Egalitarian Equivalence.}\footnote{{When a bundle $\omega$ of objects (goods or bads)  is divided, Egalitarian Equivalence calibrates an agent's absolute utility $u$ for the share $z$ as the fraction $\lambda $ of $\omega$ such that $u(z)=u(\lambda
		\omega)$; if $u$ is homogeneous of degree 1, e.g., instance additive, the
		calibrated utility is then $\frac{u(z)}{u(\omega)}$. Our calibration can be recovered if we identify the random object with a bundle $\omega$ by interpreting the probability of a particular realization as  its amount in the bundle.}} Introduced by~\cite{PaznerS1978}, it has been popular ever since in the division literature (e.g., \citealt{BramsTaylor1996, BMSY2018, Moulin2019}).} In that literature it is used to
pursue Max-min fairness, while we use it to maximize (minimize) a
utilitarian objective: the sum of {normalized} (dis)utilities.

{Normalizing} utilities by their expected value is natural but not the only
possible option. Another familiar approach is to calibrate the range of the
random utilities, from $0$ in the worst outcome to $1$ in the best:
maximizing the sum of utilities thus calibrated, known as
\emph{Relative Utilitarianism}, is the object
of recent\ axiomatic work by \cite{Dhillon1998}, \cite{DhillonMertens1999},
and \cite{BorgersChoo2017}.

We note finally that if individual (dis)utilities are measured in money and
transferable across agents, there is no need for further normalization and
fairness is achieved by cash compensations. Our division rules are useless
in that context.

\paragraph{Strategic manipulations.}

The Proportional and our Top-Heavy and Bottom-Heavy rules are fair only if
they rely on the correct profile of (dis)utilities and their expected
values. If these parameters are not objectively measurable, they must be
reported truthfully by the agents. As revelation mechanisms, our division
rules are not incentive compatible. Clearly, in the one-shot context of our
model, the only fair incentive-compatible division rule is the Equal Split,
which ignores utilities altogether.

\subsection{More relevant literature}

\label{subsect_lit} Starting with Diamond's well-known paradox \citep{Diamond1967}, the microeconomic literature on fairness under uncertainty
focuses on the trade-off between ex~post and ex~ante fairness in the context
of public decision making and discusses ways of adapting the social welfare
approach to capture this tension: notable contributions include \cite{Broome1984}, \cite{BenPorath1997}, and \cite{Gadjos2004}.

\cite{Myerson1981} initiates the discussion of fair division under
uncertainty, in the axiomatic bargaining model: there as in our model agents
are risk neutral and ex~ante fairness allows significant efficiency gains,
the same ones our rules are designed to capture.

The design of our division rules handling only a very limited amount of
information is methodologically close to the design of {prior-independent~\citep{Devanur2011} and prior-free~\citep{Hartline2001}} auctions and the
application of robust optimization to contract theory~\citep{Caroll2015}.
There as here, in contrast to the classic Bayesian approach where the
manager knows the prior distribution, either no information about the prior is
available at all or it is known that the prior belongs to a certain wide
class of distributions. And, there as here, the optimal worst-case behavior is the main
design objective.

Our model is static, yet we can interpret it as an allocation decision
taken multiple times, which justifies the use of ex~ante fairness,
i.e., fairness in expectation (see also the discussion of Example~\ref{ex_1}
in the Introduction). This interpretation links our setting and the active
current research about ``online'' resource
allocation, dealing with sequential allocative decisions when future
resources are uncertain, e.g., \cite{KarpVazirani1990}, \cite{Feldman2009}, and \cite{Devanur2019}. Our notion of the Competitive Ratio is inspired by the
competitive analysis common to this literature~\citep{BorodinYaniv2005}.

Online fair division is a very recent topic, adding fairness
concerns to the standard efficiency objectives of online resource
allocation. However, most of the papers on online fair division either
ignore the efficiency objective entirely and focus exclusively on fairness 
\citep{Benade2018}, or they consider both objectives but impose strong
simplifying assumptions on the structure of preferences~\citep{Aleksandrov2015}. The two papers closest to ours are the follow-up works by \cite{ZengPsomas} and~\cite{Gkatz2020}. The first one looks
at the fairness/efficiency trade-off when the allocation rule competes
against an increasingly adversarial nature. One of their settings (i.i.d.
goods with known distribution) reduces to our static problem, and the rule
they come up with can be seen as the Optimal fair prior-dependent rule, where
fairness is interpreted as Envy Freeness. The second paper considers a non-probabilistic setting, where the utilities are determined by an adversary, however, the sum of the utilities over periods is known to the manager. This requirement is parallel to our assumption of known expected values. Despite this similarity, characterization of optimal rules in the setting of \cite{Gkatz2020} turns out to be problematic even in the two-agent case. 


%
%


\section{The model}\label{sec2}
Definitions~\ref{def1} to \ref{def6} apply to the division of a good or a bad.
\begin{definition}\label{def1}
{A fair division problem} $\mathcal{P}=(N,\mu ,X)$ {%
		is described by the fixed set } {$N=\{1,2,\dots, n\}$}  {of agents, the
		probability distribution }$\mu$  {{over the positive orthant} $\mathbb{R}_{+}^{N}$}{, and the random variable }$X$ { in }$%
	\mathbb{R}
	_{+}^{N}$ { with distribution }$\mu $. {We always assume that
		the expectations }$\mathbb{E}_{\mu }(X_{i})$ {are bounded and positive
		for each }$i$.
\end{definition}	

We interpret $X_{i},\ i\in N$, as agent $i$'s random utility or disutility and impose no additional restriction
on the probability space or the distribution of $X$: (dis)utilities $X_{i}$
may be arbitrarily correlated across agents.

 { We write {$X_{i}^*=\frac{1}{%
		\mathbb{E}_{\mu }(X_{i})}X_{i}$} for agent $i$'s {normalized} utility or
	disutility.}  {We assume that upon
	the arrival of each object, the corresponding profile {$X^*$ of normalized
	(dis)utilities  is revealed:} it is the input of our division rules. In other
	words, the rule learns how lucky or unlucky each agent is  to receive
	the object that just appeared. }
\begin{definition}\label{def2}
	{A  (prior-dependent) division rule $\varphi$ is a collection of
		measurable mappings }   {{$\varphi ^{\mu}$ {from} $\mathbb{R}_{+}^{N}$ {to the simplex} 
			$\Delta (N)$ { of lotteries over~$N$, one for each prior distribution}}  $\mathcal{\mu}$. 
		{Given a division problem $\mathcal{P}=(N,\mu,X)$} and a
		realization}  $x^*\in 
	\mathbb{R}
	_{+}^{N}$  of the normalized (dis)utility profile $X^*$,
	agent $i$  gets the share $\varphi _{i}^\mu(x^*)$ of the
	object.
\end{definition} 	

Here ``dividing the object'' can be
interpreted either literally if the object is divisible, or as assigning
probabilistic shares, or time shares.

\begin{definition}\label{def3}
A division rule $\varphi$ is almost prior-independent (API)  if it does not depend on $\mu$, i.e., $\varphi^\mu=\varphi^{\mu'}$ for all distributions $\mu$ and\footnote{{In practice, it is reasonable to assume that the realized vector of (dis)utilities $X$ is observed directly, while the normalized (dis)utilities $X_i^*=\frac{X_i}{\mathbb{E}_\mu(X_i)}$ are derived from it. Hence, in order to apply an API rule one still needs to know the expected values.}} $\mu'$.  For API rules we will drop the superscript~$\mu$.	
\end{definition}
{We focus on rules that treat agents similarly, i.e.,  satisfy symmetry.}
\begin{definition}\label{def4}
{A rule is symmetric if a permutation of the agents permutes their shares accordingly.  {Formally, for any distribution $\mu$,  vector $x\in \mathbb{R}_+^N$, agent $i\in N$, and any permutation $\pi$ of $N$, we have $\varphi_i^\mu(x)=\varphi_{\pi(i)}^{\pi(\mu)}(\pi(x))$, where $\pi(x)$ and $\pi(\mu)$ are obtained from $x$ and $\mu$ by permuting coordinates: $\pi(x)_{\pi(j)}=x_j$, $j\in N$, and $\pi(\mu)(\pi(A))=\mu(A)$, for any measurable set $A$ and $\pi(A)=\{\pi(y)\,:\,y\in A\}$.}}
\end{definition}

The fairness constraint  sets a lower (resp. upper) bound on every agent's expected utility (resp. disutility).
\begin{definition}\label{def5}
{The division rule }$\varphi$
{ guarantees Fair Share (FS) if every agent's expected (dis)utility
	is at least (at most) }$\frac{1}{n}${-th of his expected (dis)utility
	for the entire object.} 
{If the object is a good, this means  {for each division problem $\mathcal{P}$ and each agent $
	i\in N$}},
{ {\begin{equation}
\mathbb{E}_{\mu }\left(\varphi _{i}^\mu(X^*)\cdot X^*_i\right)\geq \frac{1}{n}. \label{12}
\end{equation}}}
{The inequality is reversed if we divide a bad.}
\end{definition}

{We define expected social welfare (the expected social cost in the case of a bad) as the expected sum of normalized (dis)utilities
\begin{equation}\label{eq_SW}	
S(\varphi,\mathcal{P})=\mathbb{E}_{\mu}\left(\sum_{i\in N}\varphi _{i}^\mu(X^*)\cdot X_i^*\right).
\end{equation}
Our design goal, conditional upon satisfying Fair Share, is to maximize $S(\varphi,\mathcal{P})$
in the case of a good, or to minimize this quantity} in the case of a bad. 

{Both of our design objectives~\eqref{12} and~\eqref{eq_SW} are invariant with respect to rescaling of (dis)utilities. Since our division rules also depend on normalized (dis)utilities, in the rest of the paper we can restrict attention to those problems where $X$ and $X^*$ coincide.}
\begin{definition}\label{def6}
{We call the problem} $\mathcal{P}$ {%
	normalized if }$\mathbb{E}_{\mu }(X_{i})=1$ {for all} $i\in N$.
\end{definition}

All proofs are given for normalized problems and
extend automatically to general problems by replacing everywhere $X_{i}$ by $%
X_{i}^*=\frac{1}{\mathbb{E}_{\mu }(X_{i})}X_{i}$.

\textit{Notation. }  Throughout the paper we use the following notation. For a vector $x\in \mathbb{R}^N$ and a subset $M\subseteq N$, the sum of coordinates over $M$ is  denoted by $x_M=\sum_{j\in M}x_{j}$.  By $e^M\in \mathbb{R}^N$, we denote the indicator vector of a subset $M\subset N$, i.e., $e_{i}^{M}=1$ if $i\in M$ and $e_{i}^{M}=0$ if $i\notin M$. Finally $x\gg y$ means $x_{i}>y_{i}$ for all $i$.

\subsection{Three benchmark  {API rules}}\label{subsect_benchmark_rules}
The \textit{Equal Split rule}, $\varphi^{es}(x)=\frac{1}{n}e^{N}$ for all $%
x $, is the simplest {API} rule of all, and it implements
Fair Share.  {Not surprisingly, its {efficiency} is poor.} 

{On the other extreme, we have the \textit{Utilitarian rule} $\varphi^{ut}(x)=\frac{e^M}{|M|}$, where $M=\{i\in N:\, x_i=\max_{j\in N} x_j\}$ for a good and $M=\{i\in N:\, x_i=\min_{j\in N} x_j\}$ for a bad. This rule achieves the optimal welfare level  by allocating the object among agents with highest (lowest) normalized (dis)utilities. However, it drastically violates FS: in a two-agent normalized problem with a good, where  $X=(1,1+\varepsilon)$ with probability $(1-\varepsilon)$ and $(1,\varepsilon)$ with probability $\varepsilon$, the expected utility of the first agent $\mathbb{E} \left(\varphi^{ut}_1(X)X_1\right)=\varepsilon $ is below his fair share of $\frac{1}{2}$  for any $\varepsilon\in\left(0,\frac{1}{2}\right)$.}

{A natural compromise between these two rules is the \textit{Proportional rule}, which is defined as follows:
	$$
	\mbox{for a good:} \ \ \varphi _{i}^{pro}(x)=\frac{x_{i}}{x_{N}}, \ \  \mbox {if} \ x\neq 0,
	\quad\mbox{and } \ \ \varphi^{pro}(0)=\frac{e^N}{n}
	$$
	$$
	\mbox{for a bad:}  \ \  \varphi _{i}^{pro}(x)=
	\frac{\frac{1}{x_{i}}}{\sum_{j\in N}\frac{1}{x_{j}}},  \ \  \mbox{if } \ x\gg 0,
	\quad\mbox{and } \ \ \varphi ^{pro}(x)=\frac{e^M}{|M|}, \ \mbox{where } M=\{i\in N: x_i=0\}\ne\varnothing.
	$$}
{The next  proposition shows that the Proportional rule also guarantees FS and 
	generates a higher {social welfare} than $\varphi^{es}$ in the following strong ex~post sense.}
\begin{definition}\label{def7}
{Fix two {API} rules }$\varphi$ and $\psi$.  We say that 
	$\varphi${ dominates }$\psi${ if it always
		generates ex~post (for every realization of the {normalized} utilities) at least
		as much {social welfare}, and sometimes strictly more. Formally, in the case of a good,}%
	\begin{equation}
	\sum_{i\in N}\psi _{i}(x)\cdot x_{i}\leq \sum_{i\in N}\varphi
	_{i}(x)\cdot x_{i}\mbox{ for all }x\in 
	\mathbb{R}
	_{+}^{N}\mbox{ , with a strict inequality for some }x.  \label{17}
	\end{equation}
	{In the case of a bad, the inequalities are reversed.}
\end{definition}
\begin{proposition}\label{prop1}
	{{The Proportional rule guarantees Fair Share and dominates the Equal Split both for a good and for a bad.}}
\end{proposition}	
\begin{proof}[\textit{Proof for a good.}] Suppose that $%
\mathcal{P}$ is normalized.  {To prove FS,} apply the Cauchy--Schwartz inequality to the
two variables $\frac{X_{i}^{2}}{X_{N}}$ and $X_{N}$:
$\mathbb{E}_{\mu }\left(\frac{X_{i}^{2}}{X_{N}}\right)\cdot \mathbb{E}_{\mu }(X_{N})\geq
(\mathbb{E}_{\mu }X_{i})^{2}.$ 
Now the left-most expectation is simply $\mathbb{E}_{\mu }(\varphi
_{i}^{pro}(X)\cdot X_{i})$, agent $i$'s expected utility, while by the
normalization the other two terms are respectively $n$ and $1$.   {Next, the weak domination condition~(\ref{17}) reads as $\frac{x_N}{n}\leq \frac{\sum_{i\in N} x_i^2}{x_N}$, which is equivalent to the inequality between arithmetic and  quadratic means: $\frac{x_N}{n}\leq \sqrt{\frac{\sum_{i\in N} x_i^2}{n}}$. If  $x$ is not proportional to $e^N$, the inequality becomes strict.}

\smallskip

\noindent\textit{Proof for a bad.} Agent $i$'s expected
utility under $\varphi ^{pro}$ is now $\mathbb{E}_{\mu }\left(\varphi _{i}^{pro}(X)\cdot X_{i}\right)=\mathbb{E}_{\mu }\left(\frac{1%
}{\sum_{j\in N}\frac{1}{X_{j}}}\right)=\frac{1}{n}\mathbb{E}_{\mu }(\widetilde{X}),$ 
where  {$\widetilde{x}$ denotes the harmonic mean of the $x_{i}$'s.}  {FS then}
follows from the inequality $\widetilde{x}\leq \frac{x_N}{n}$  {between harmonic and arithmetic means.    The weak domination condition~(\ref{17}) boils down to the same inequality, which, as in the case of a good, becomes strict whenever $x$ is not proportional to $e^N$.} 
\end{proof}

{For a good, the ratio } $\frac{\sum_{i\in N}\varphi _{i}^{pro}(x)\cdot x_{i}}{\sum_{i\in
		N}\varphi _{i}^{es}(x)\cdot x_{i}}$ can be as high as $n$, while  for a bad the ratio
$\frac{\sum_{i\in N}\varphi _{i}^{es}(x)\cdot
	x_{i}}{\sum_{i\in N}\varphi _{i}^{pro}(x)\cdot x_{i}}$ can be
arbitrarily large.  {For example, } take $x=e^{\{1\}}$ for a good and $x=\varepsilon e^{\{1\}}+e^{N\diagdown \{1\}}$, where $\varepsilon $ is arbitrarily small, for a bad.

{One can try to  { achieve greater efficiency than the Proportional rule} does }by assigning probabilities to agents in proportion
(or inverse proportion) to some strictly higher power $q>1$ of their {normalized}
(dis)utilities, but such rules fail FS.\footnote{
	Suppose that we divide a good and $\mu $ picks, for each $i\geq 2$, the vector $%
	x^{i}=e^{\{1\}}+(n-1)e^{\{i\}}$ with probability $\frac{1}{n-1}$. Then the expected
	utility of agent $1$ is $\frac{1}{1+(n-1)^{q}}$, which is below $\frac{1}{n}$ for $%
	n\geq 3$. The proof for a bad is similar.}
{In the next two sections we construct fair {API} rules with higher performance than $\varphi^{pro}$. }

\section{Goods: The family of undominated {API} rules}\label{sec3}

Our first main result (Theorem~\ref{th1} below)  {describes the set of undominated  {API} rules in the sense of Definition~\ref{def7}.}

\subsection{ {Characterizing fairness for a good}} The key step toward Theorem~\ref{th1} characterizes the restriction imposed by Fair
Share on any {API} rule $\varphi $. Given a vector $x$ in $%
\mathbb{R}
_{+}^{N}$, we write its arithmetic average as $\overline{x}=\frac{x_N}{n}$.

\begin{proposition}\label{prop2}
 \textit{ {A  {symmetric}} {API} rule }$\varphi $\textit{\
	dividing a good satisfies Fair  {Share} if and only if there exists a number }$%
\theta, \ 0\leq \theta \leq 1$\textit{, such that}%
\begin{equation}
\varphi _{i}(x)\geq \max \left\{\frac{1}{n}+\frac{\theta }{n-1}\left(1-\frac{\overline{%
		x}}{x_{i}}\right), \ 0\right\}\mbox{ for all }i\in N\mbox{ and }x\in 
\mathbb{R}
_{+}^{N}  \label{10}
\end{equation}%
\textit{(where we use }$\frac{1}{0}=+\infty $\textit{)}.
\end{proposition}
\begin{proof}[Proof of the ``\textbf{if}'' statement.]
Assume that the division rule $\varphi $
for a good satisfies (\ref{10}); then%
\[
\varphi _{i}(x)\cdot x_{i}\geq \frac{1}{n}x_{i}+\frac{\theta }{n-1}(x_{i}-%
\overline{x}) \ \ \mbox{ for all } \ x.
\]%
For an arbitrary normalized problem $\mathcal{P}$ (Definition~\ref{def6}), we have $%
\mathbb{E}_{\mu }(X_{i}-\overline{X})=0$ and  the inequality~(\ref{12}%
) follows. 

\smallskip
\noindent\textit{Proof of the ``\textbf{only if}'' statement.}
Assume that the rule $%
\varphi $ satisfies Fair Share and define the real-valued function $f(x)=\varphi
_{1}(x)\cdot x_{1}$. By the symmetry of $\varphi$, we get  $f(e^{N})=\frac{1}{n}$.
Consider a convex combination in $%
\mathbb{R}
_{+}^{N}$, with an arbitrary number of terms, such that $\sum_{k=1}^{K}\mu
_{k}y^{k}=e^{N}$. The problem $\mathcal{P}$ in which $X=y^{k}$ with
probability $\mu _{k}$ is normalized and FS implies%
\[
\sum_{k=1}^{K}\mu _{k}f(y^{k})\geq \frac{1}{n}=f(e^{N}).
\]%
 {Recall that the convexification $g$ of $f$ is the pointwise-maximal convex function such that $g(x)\leq f(x)$ for all $x$. Alternatively, $g(x)$ can be represented as 
	\begin{equation}\label{eq_convexification}
g(x)=\inf \Big\{\sum_{k=1}^{K}\mu
_{k}f(y^{k})\Big\},
	\end{equation}
	where the infimum is over all\footnote{ {By the Caratheodory theorem \citep[Theorem 17.1]{Rockafellar1970},  it is enough to take the infimum in~(\ref{eq_convexification}) over convex combinations with at most $n+1$ points. This allows us to strengthen the ``only if'' part of Proposition~\ref{prop2}:  the bound~(\ref{10}) holds if the rule $\varphi$ satisfies FS in all problems with at most $n+1$ goods.}} convex combinations such that $\sum_{k=1}^{K}\mu
	_{k}y^{k}=x$; e.g.~\citet[Proposition 1.1]{Laraki2004}.}

The above inequality implies $g(e^{N})\geq f(e^{N})$ and the opposite
inequality is true by the definition of $g$, and so $g(e^{N})=f(e^{N})$. Because $g$
is convex and finite at $e^{N}$ there exists a vector $\alpha \in 
\mathbb{R}
^{N}$ supporting its graph at $(e^{N},g(e^{N}))$, i.e., such that for all $%
x\in 
\mathbb{R}
_{+}^{N}$: $
g(x)\geq g(e^{N})+\alpha \cdot (x-e^{N})$.  {Therefore,
	$$f(x)=\varphi
	_{1}(x)\cdot x_{1}\geq \frac{1}{n}+\alpha \cdot (x-e^{N}). 
	$$}
\begin{figure}[h!]
	\centering
	\vskip 0.0 cm
	{\includegraphics[width=15cm, clip=true, trim=0cm 1.0cm 0cm 4.5cm]{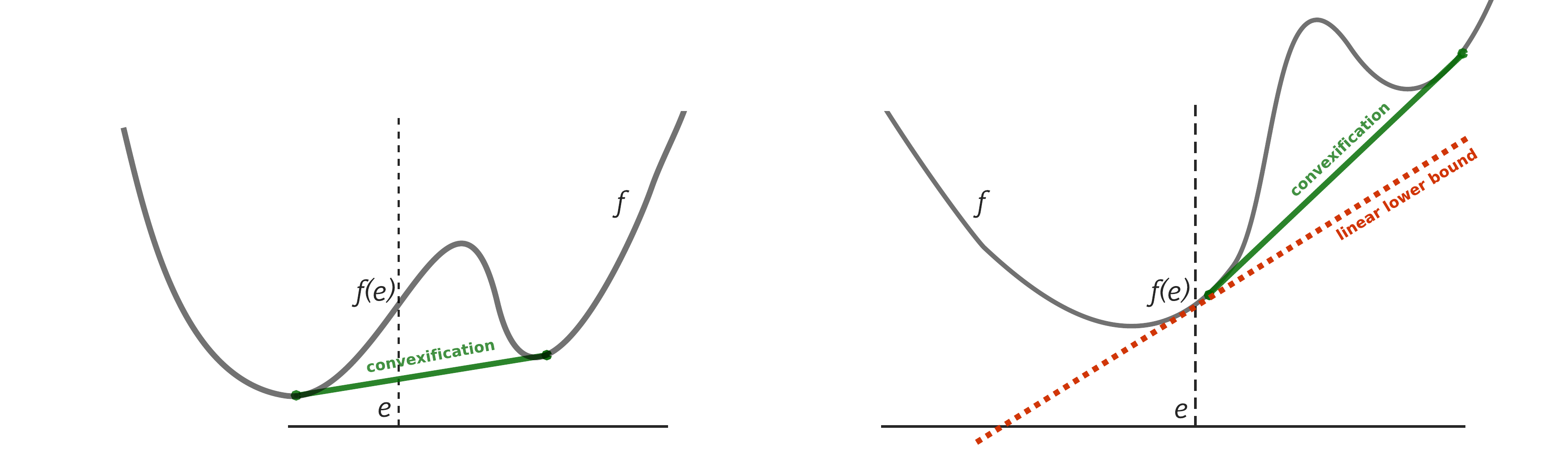}}
	\caption{The geometric intuition behind the proof of Proposition~\ref{prop2}. Right figure: the convexification of a function $f$ coincides with $f$ at $x=e$ if the graph of $f$ is supported by a linear function. The left figure illustrates the necessity of this condition.}
	\label{fig1}
\end{figure}

Apply the inequality above to $x=\lambda e^{N}$ for any $\lambda >0$. By the
symmetry of $\varphi $ we get%
\[
\frac{1}{n}\lambda \geq \frac{1}{n}+(\lambda -1)\alpha \cdot e^{N}%
\mbox{ for
	any }\lambda >0. 
\]%
 {Pushing $\lambda$ to $+\infty$ and to $+0$ yields two opposite inequalities: $\frac{1}{n}\geq \alpha \cdot e^{N}$ and $\alpha \cdot e^{N}\geq \frac{1}{n}$, respectively.
	Therefore, $\alpha \cdot e^{N}=\frac{1}{n}$ and $\varphi _{1}(x)\cdot
	x_{1}\geq \alpha \cdot x$ for all $x$.}

Again, the symmetry of $\varphi $ implies that we can take $\alpha _{j}=\alpha
_{i}$ for all $i,j\geq 2$. Indeed, if $x^{\prime }$ results from $x$ by
permuting coordinates $i$ and $j$, we have%
\[
\varphi _{1}(x)\cdot x_{1}=\varphi _{1}(x^{\prime })\cdot x_{1}\geq \frac{1}{
	2}(\alpha \cdot x+\alpha \cdot x^{\prime })=\widetilde{\alpha }\cdot x, 
\]%
where $\widetilde{\alpha }_{i}=\widetilde{\alpha }_{j}$ and $\widetilde{
	\alpha }\cdot e^{N}=\frac{1}{n}$ are preserved.

 {Set $\beta =-\alpha _{i}$ for all $i\geq 2$. {Since $\varphi_1(x)$ belongs to $[0,1]$, we obtain the following chain of inequalities:   $x_{1}\geq \varphi _{1}(x)\cdot x_{1}\geq \alpha
		_{1}x_{1}-\beta x_{N\diagdown \{1\}}$. Keeping $x_1$ bounded and pushing $x_{N\diagdown \{1\}}$ to infinity, we get that this chain of inequalities can  hold only if  $\beta \geq 0$.} Combining this with $\alpha \cdot e^{N}=%
	\frac{1}{n}$ we see that 
	\[
	\varphi _{1}(x)\cdot x_{1}\geq \alpha \cdot x=\frac{1}{n}x_{1}+\beta 
	{\large (}(n-1)x_{1}-x_{N\diagdown \{1\}}{\large )}. 
	\]%
	Changing the parameter $\beta $ to $\delta =n\beta $}  gives%
 {
	\[
	\varphi _{i}(x)\geq \frac{1}{n}+\delta \left(1-\frac{
		\overline{x}}{x_{i}}\right)\mbox{ for all }x\in 
	\mathbb{R}
	_{+}^{N} 
	\]
	and $i=1$. For the remaining agents $i\in N\setminus\{1\}$, this bound with the same~$\delta$ follows by the symmetry of $\varphi$: if $x$ and $x^{\prime}$ differ by permuting coordinates of $1$ and $i$, then $\varphi_1(x^{\prime})=\varphi_i(x)$.}

It remains to find the bounds on  {$\delta $} derived from the fact that $%
\varphi (x)$ is in $\Delta (N)$. For all $x\gg 0$, the above inequality 
and $\varphi (x)\geq 0$ imply 
\begin{equation}
\sum_{i\in N}\max \left\{\frac{1}{n}+\delta \left(1-\frac{\overline{x}}{x_{i}}%
\right),0\right\}\leq 1\mbox{ for all }x\in 
\mathbb{R}
_{+}^{N},  \label{13}
\end{equation}%
which is equivalent to the following property:%
\[
\mbox{for all }M\subseteq N: \ \sum_{i\in M}\left(\frac{1}{n}+\delta \left(1-%
\frac{\overline{x}}{x_{i}}\right)\right)=|M|\left(\frac{1}{n}+\delta \right)-\delta \overline{x}%
\left(\sum_{i\in M}\frac{1}{x_{i}}\right)\leq 1 \ \mbox{ for all }x\in 
\mathbb{R}
_{+}^{N}. 
\]%
 {By the inequality between harmonic and arithmetic means, $\frac{\sum_{i\in M} x_i}{|M|}\geq\frac{|M|}{\sum_{i\in M}\frac{1}{x_i}}$. Since $\overline{x}\geq \frac{1}{n}\sum_{i\in M} x_i$,	
	the infimum of $\overline{x}(\sum_{i\in M}\frac{1}{x_{i}})$ is $\frac{|M|^{2}%
	}{n}$, which is achieved for any $x$ parallel to $e^{M}$;} therefore,%
\[
|M|\left(\frac{1}{n}+\delta \right)\leq 1+\delta \frac{|M|^{2}}{n}\Longleftrightarrow
\left(1-\frac{|M|}{n}\right)(\delta |M|-1)\leq 0 
\]%
and we conclude that $\delta \leq \frac{1}{n-1}$. This gives the desired
inequality (\ref{10}) by setting $\theta =(n-1)\delta $.
\end{proof}

\subsection{ {Undominated rules for a good: The Top-Heavy family}}\label{subsect3.2}
Armed with Proposition~\ref{prop2}, we can now easily identify the undominated 
	{API} division rules (Definition~\ref{def7}) satisfying FS for goods.

For any $x\in 
\mathbb{R}
_{+}^{N}$, we\ write {$(x^{(1)},\ldots ,x^{(n)})$} for the
order statistics\footnote{%
	The vector with the same set of coordinates as $x$, rearranged in increasing
	order.} of $x$, and  {$\tau (x)=\{i\in N|x_{i}=x^{(n)}\}$} for the set of agents with the largest utility.

{We fix $\theta, \ 0<\theta \leq 1$,  and define the \textit{Top-Heavy} rule $\varphi^\theta$
	by placing as much weight on the agents from $\tau(x)$ as inequalities (\ref{10}) permit.}
\begin{definition}\label{def8}
{For $0<\theta\leq 1$, the Top-Heavy (TH) rule $\varphi^\theta$
	is given by}
\begin{equation}
\varphi _{i}^{\theta }(x)=\left\{
\begin{array}{cc}
{\displaystyle \max \left\{\frac{1}{n}+\frac{\theta }{n-1}\left(1-\frac{%
		\overline{x}}{x_{i}}\right),\ 0\right\},} &   i\in N\diagdown \tau (x)\\
{\displaystyle\frac{1}{|\tau (x)|}\left(1-\sum_{j\in
		N\diagdown \tau (x)}\varphi _{j}^{\theta }(x)\right),} & i\in
\tau (x) 
\end{array}
\right..
\label{16}
\end{equation}
 {{Thus  all agents except those with the highest values receive  the  share $\varphi _{i}^{\theta }(x)$, which  is equal to the lower bound~(\ref{10}), while the agents with the highest values equally split the rest.}}
\end{definition}

Inequality (\ref{13}) guarantees that the shares received by the agents with the highest values are non-negative. It also implies that the $i$-sequence of shares $\varphi
_{i}^{\theta }(x)$ is co-monotonic with that of utilities $x_{i}$,  {i.e., $x_k\geq x_i$ implies\footnote{%
		This is clear if we compare the shares of two agents $i,k$ outside $\tau (x)$; if $i\notin \tau (x)$ and $k\in \tau (x)$, inequality $\varphi _{i}^{\theta
		}(x)\leq \varphi _{k}^{\theta }(x)$ is 
		\[
		|\tau (x)|\varphi _{i}^{\theta }(x)+\sum_{j\in N\diagdown \tau (x)}\varphi
		_{j}^{\theta }(x)\leq 1, 
		\]%
		which follows from $\varphi _{i}^{\theta }(x)=\max \{\frac{1}{n}+\delta 
		{\large (}1-\frac{\overline{x}}{x_{i}}{\large ),0\}}$ and (\ref{13}).} $\varphi
	_{k}^{\theta }(x)\geq\varphi
	_{i}^{\theta }(x)$.}

The rule $\varphi ^{\theta }$ converges to Equal Split when $\theta $ goes
to zero, but Equal Split is clearly dominated by \textit{any} rule $\varphi
^{\theta }$ for $\theta >0$. This is why we excluded $0$ from the range of $%
\theta $.

Note that the  {discontinuity of $|\tau(x)|$ implies that} for $n\geq 3$, all rules $\varphi ^{\theta }$ are discontinuous at any $x$ where
at least  {two agents, but not all, have the highest utility {($x^{(1)}<x^{(n-1)}=x^{(n)}$).} For two agents, the TH rule is continuous.}

\begin{example}[the TH rule  $\varphi^1$ for two agents]
For two-agent problems, the rule $\varphi ^{1}$ has a simple expression. By
symmetry it is enough to define it when $x_{1}\leq x_{2}$:%
\begin{equation}
\varphi ^{1}(x)=\left\{\begin{array}{cc}
(0,1), & \ \frac{x_{1}}{x_2}\leq \frac{1}{2}\\
\left(1-\frac{
	x_{2}}{2x_{1}}, \ \frac{x_{2}}{2x_{1}}\right), &  \frac{1}{2}\leq \frac{x_{1}}{x_2}\leq 1 
\end{array}\right..   \label{23}
\end{equation}
The dependence of $\varphi_1^1$ on $\frac{x_1}{x_2}$ is depicted in Figure~\ref{fig2}.
\begin{figure}
	\centering
	{\includegraphics[width=6cm, clip=true]{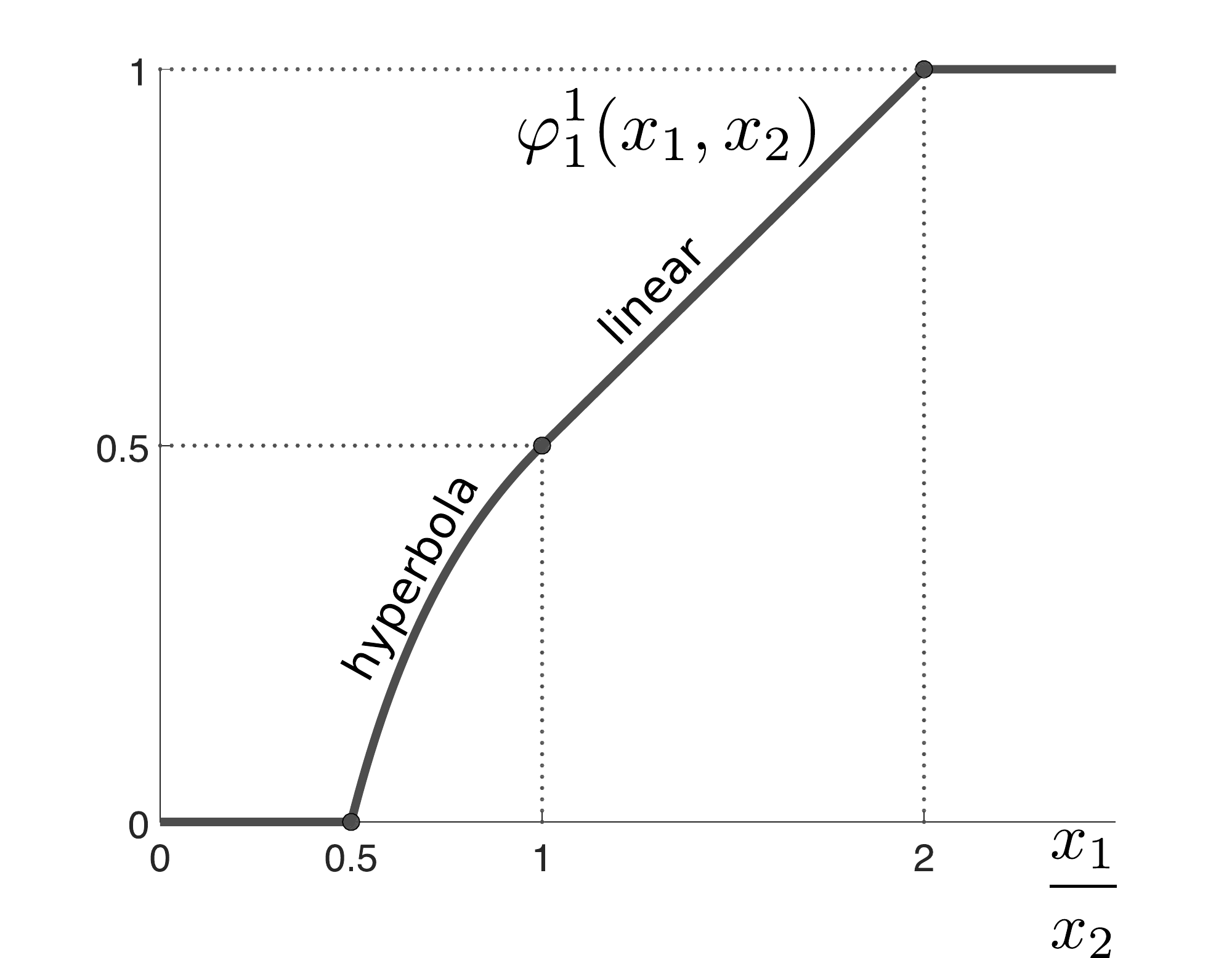}}
	\caption{The amount of a good received by the first agent under the TH rule $\varphi^1$ for two agents as a function of the ratio $t=\frac{x_1}{x_2}$. If the ratio is below $\frac{1}{2}$ or above $2$, the TH rule coincides with the Utilitarian one, which gives the whole good to an agent with the highest value. If the {normalized utilities} are closer, both agents receive a non-zero amount of the good:				
		$\varphi_1=1-\frac{1}{2t} $ on $\left[\frac{1}{2},1\right]$ and $\varphi_1=\frac{1}{2}t $ on $\left[1,2\right]$.}
	\label{fig2}
\end{figure}
\end{example}
\pagebreak
\begin{theorem}[for goods]\label{th1} \ 
\begin{enumerate}
	\item \textit{ For any }$n\geq 2$,\textit{ every  {symmetric} {API}  rule satisfying Fair Share is dominated by, or equal to, one Top-Heavy rule }$%
		\varphi ^{\theta }$\textit{,}  $0<\theta \leq 1$.	
\item \textit{If }$n=2$,\textit{ the Top-Heavy rule }$\varphi ^{1}$%
		\textit{ dominates every other  Top-Heavy rule $\varphi^\theta$, $0<\theta < 1$.}
\item\textit{If  }$n\geq 3$, \textit{the Top-Heavy rules
			$%
			\varphi ^{\theta }$\textit{,}  $0<\theta \leq 1$\textit{,} are undominated.}
\item \textit{{For $n\geq 3$,} the Proportional rule is dominated by the Top
		Heavy rule }$\varphi ^{\frac{n-1}{n}}$\textit{, but not by any other rule }$%
	\varphi ^{\theta }$\textit{.}
\end{enumerate}	
\end{theorem}
\begin{proof}[Proof of statement $i)$] Fix an {API} rule $\varphi 
$ satisfying FS. There is a $\theta, \ 0\leq \theta \leq 1$, such that the
inequalities (\ref{10}) hold for all $i$ and $x$ (Proposition~\ref{prop2}). If $\theta
=0$, our rule is Equal Split, which we already noticed is dominated by each
rule $\varphi ^{\theta }$. If $\theta >0$, these inequalities imply $\varphi
_{i}(x)\geq \varphi _{i}^{\theta }(x)$ for all $x$ and all $i\notin \tau (x)$%
. Hence $(\varphi _{i}(x)-\varphi _{i}^{\theta }(x))x_{i}\leq (\varphi
_{i}(x)-\varphi _{i}^{\theta }(x))x^{(n)}$ for all $i\notin \tau (x)$.
Summing up these inequalities and adding $\sum_{j\in \tau (x)}(\varphi
_{i}(x)-\varphi _{i}^{\theta }(x))x_{j}$ on both sides gives the desired
weak inequalities in (\ref{17}). If none of the inequalities in (\ref{17})
is strict, we deduce $\varphi _{i}(x)=\varphi _{i}^{\theta }(x)$ for all $x$
and all $i\notin \tau (x)$ such that $x_{i}>0$. If there is some $i$ such that $%
x_{i}=0$ and $\varphi _{i}(x)>0$ (while $\varphi _{i}^{\theta }(x)=0$) then $%
\varphi (x)$ has less weight to distribute on $\tau (x)$ than $\varphi
^{\theta }$, contradicting our assumption. Because $\varphi $ is symmetric,
we conclude  that $\varphi (x)=\varphi ^{\theta }(x)$. 

\smallskip
\noindent\textit{Proof of statement $ii)$}  Fix $\theta <\theta ^{\prime }$ and $%
x_{1}\leq x_{2}$. The formula (\ref{16}) implies $\varphi _{1}^{\theta
}(x)\geq \varphi _{1}^{\theta ^{\prime }}(x)$ because the coefficient of $%
\theta $ in $\varphi _{1}^{\theta }(x)$ is $\frac{1}{2}(1-\frac{x_{2}}{x_{1}}%
)\leq 0$.   {Hence, under $\varphi^{\theta{^\prime}}$, the low-value agent~1 receives the good with lower probability than under $\varphi^{\theta}$. This yields inequality (\ref{17}) and, 
	for $\frac{1}{2}< \frac{x_1}{x_2}<1$, it is strict.} Thus $\varphi ^{1}$ dominates $\varphi ^{\theta }$ for $\theta <1$. 
Note that this argument does not extend to the case $n\geq 3$ because
if agent $i$'s utility is neither the smallest nor the largest, the sign of
the coefficient of $\theta $ in $\varphi _{i}^{\theta }(x)$ is ambiguous.

\smallskip
\noindent\textit{Proof of statement $iii)$}  
We check now that no TH rule $\varphi ^{\theta }$ dominates another TH rule $%
\varphi ^{\theta ^{\prime }}$. Assume that $0<\theta <\theta ^{\prime }$ and
consider first the profile $x_{i}=\frac{3}{4}$ if $i\neq n$ and $x_{n}=1+%
\frac{n-1}{4}$. Then $\overline{x}=1$ and all coordinates of $\varphi
_{i}^{\theta }(x)$ and $\varphi _{i}^{\theta ^{\prime }}(x)$ are strictly
positive. Compute $\varphi _{i}^{\theta }(x)-\varphi _{i}^{\theta ^{\prime
}}(x)=\frac{\theta ^{\prime }-\theta }{3(n-1)}>0$ for all $i\neq n$, and so 
$\varphi ^{\theta ^{\prime }}$ generates more surplus at $x$ than $\varphi
^{\theta }$.

To show an instance of the reverse comparison, we choose $x_{1}=\frac{\theta }{3}$,  $x_{i}=1+\frac{\frac{3}{4}-\frac{\theta }{3}%
}{n-2}$  for  {$2\leq i\leq n-1,$} and $x_{n}=\frac{5}{4}.$
Thus $\overline{x}=1$ and $\overline{x}<x_{i}<x_{n}$ for  {$2\leq i\leq n-1$.}
This implies $\varphi _{1}^{\theta }(x)=\varphi _{1}^{\theta ^{\prime
}}(x)=0 $, $\varphi _{i}^{\theta }(x)<\varphi _{i}^{\theta ^{\prime }}(x)$,
and $\varphi _{n}^{\theta }(x)>\varphi _{i}^{\theta ^{\prime }}(x)$.

\smallskip
\noindent\textit{Proof of statement $iv)$}  {In the proof of statement $i)$,} we showed that the rule $%
\varphi $ is dominated by $\varphi ^{\theta }$ if it satisfies inequalities (%
\ref{10}).  {Thus the rule $\varphi ^{pro}$ is dominated by the TH rule $%
	\varphi ^{\theta}$ if and only if for all $x\in 
	\mathbb{R}
	_{+}^{N}$  and $i\in N$ we have%
	$$\frac{x_{i}}{x_{N}}\geq \frac{1}{n}+\frac{\theta}{n-1}\left(1-\frac{\overline{x}}{%
		x_{i}}\right)\Longleftrightarrow \frac{x_{i}}{x_{N}}+\frac{\theta \cdot x_{N}}{n(n-1)x_i}\geq 
	\frac{1}{n}+\frac{\theta}{n-1}.$$
	By the inequality between arithmetic and geometric means, $\frac{x_{i}}{x_{N}}+\frac{\theta \cdot x_{N}}{n(n-1)x_i}\geq 2\sqrt{\frac{\theta}{n(n-1)}}$ and this lower bound is attained on $x$ such that $\frac{x_i}{x_N}=\sqrt{\frac{\theta}{n(n-1)}}$. Therefore, $\varphi ^{pro}$ is dominated by $
	\varphi ^{\theta}$  if and only if $2\sqrt{\frac{\theta}{n(n-1)}}\geq \frac{1}{n}+\frac{\theta}{n-1}$. We see that the geometric mean of $\frac{1}{n}$ and $\frac{\theta}{n-1}$ exceeds their arithmetic mean, which is only possible if the two means coincide with $\frac{1}{n}$ and $\frac{\theta}{n-1}$, respectively. Thus $\varphi ^{pro}$ is dominated by $\varphi ^{\theta}$ only for $\theta=\frac{n-1}{n}$.} \end{proof}

%
%
%
%

\section{Bads: The unique undominated {API} rule}\label{sec4}

{We adapt the approach developed in the previous section in order to characterize the undominated (Definition~\ref{def7}) {API} fair rules for a bad. }

{Surprisingly, in this case the dominating rule is  unique even for $n\geq 3$.}

\subsection{ {Characterizing fairness for a bad}}
We state the counterpart of Proposition~\ref{prop2} for a bad. The proof is in Appendix~\ref{appA}.

\begin{proposition}\label{prop3}
\textit{ {A  {symmetric}} {API} rule }$\varphi $\textit{
	dividing a bad satisfies Fair Share if and only if there exists a number }$%
\theta, \ 0\leq \theta \leq 1$\textit{, such that}%
\begin{equation}
\varphi _{i}(x)\leq \min \left\{\frac{1}{n}+\frac{\theta }{n-1}\left(\frac{\overline{x}%
}{x_{i}}-1\right),\ 1\right\}\mbox{ for all }i\in N\mbox{ and }x\in 
\mathbb{R}
_{+}^{N}  \label{22}
\end{equation}%
\textit{(where we set $\frac{1}{0}=+\infty $)}.
\end{proposition}

\subsection{ {The undominated  Bottom-Heavy rule for a bad}}

{We can now use inequality (\ref{22})
	to construct, as in the previous section, the canonical \textit{Bottom-Heavy}
	rule $\varphi ^{1}$, which corresponds to $\theta=1$. 
	The construction relies on the same
	order statistics {$(x^{(1)},\ldots ,x^{(n)})$}, but is slightly more involved. We write $%
	\sigma (x;t)=\{i\in N\mid x_{i}=x^{(t)}\}$ (and so $\sigma (x;n)=\tau (x)$) and
	use the convention  {$x^{(0)}=-\infty$ and $\sigma(x,t)=\varnothing $ for $t>n$.}
	
	The BH rule places as much weight on the smallest
	disutilities as permitted by (\ref{22}). For $\theta=1$,  the right-hand side of (\ref{22}) simplifies to
	$\frac{1}{n}+\frac{\theta}{n-1}(\frac{\overline{x}}{x_{i}}-1)=\frac{1%
	}{n(n-1)}\frac{x_{N\diagdown \{i\}}}{x_{i}}$ and we get the following expression.

	\begin{definition}\label{def9}
The Bottom-Heavy (BH) rule $\varphi^1$ is defined  by 
			\begin{equation}
			\varphi _{i}^{1}(x)=\left\{
			\begin{array}{cc}
			{\displaystyle  \frac{1}{n(n-1)}\frac{x_{N\diagdown \{i\}}}{x_{i}},  } &   i: \ x_{i}\leq x^{(\widetilde{t})} \\
			{\displaystyle \frac{1}{|\sigma (x;\widetilde{t}+1)|}\left(%
				1-\frac{1%
				}{n(n-1)}\sum_{i:\,x_{i}\leq x^{(\widetilde{t})}}\frac{x_{N\diagdown \{i\}}}{x_{i}}\right),} & i\in\sigma (x;\widetilde{t}+1)\\
			0, & \mbox{otherwise}
			\end{array}
			\right.,
			\label{18}
			\end{equation}
where $\widetilde{t}$ is the maximal $t=0,1,2,\ldots,n$ such that $\frac{1%
			}{n(n-1)}\sum_{i:\,x_{i}\leq x^{(t)}}\frac{x_{N\diagdown \{i\}}}{x_{i}}\leq 1$.
		
			 {In other words, agents are weakly ordered by their values and the longest possible prefix of low-value agents (permitted by the feasibility condition $\varphi^1\in \Delta(N)$) receives shares equal to the upper bound~(\ref{10}); agents next to that prefix  split the rest equally,  and all others get nothing.}
	\end{definition}
	
	Note that for all vectors $x$ except those parallel to $e^{N}$ we have  $\frac{1}{n(n-1)}\sum_{i\in N}\frac{x_{N\diagdown \{i\}}}{x_{i}}>1$ and thus $\widetilde{t}\leq n-1$. Indeed, 
	the minimum of $\sum_{i\in N}\frac{x_{N\diagdown \{i\}}}{x_{i}}$ over $\mathbb{R}_{+}^{N}$ is $n(n-1)$, and it is achieved by any $x$ parallel to $e^{N}$, and only
	by those: for such a vector, $\widetilde{t}=n$ and $\varphi^{1}(x)=\frac{e^N}{n}.$
	
	If\, $\widetilde{t}=0$ the only agents with a positive share are those in $%
	\sigma (x;1)$, who have the smallest disutility, and so $\varphi ^{1}$
	selects an optimal utilitarian allocation.}

{Symmetrically to the case of goods,} the sequence of shares $\varphi _{i}^{1}(x)$ is anti-monotonic to the
sequence of disutilities $x_{i}$.

\begin{example}[the BH rule  $\varphi^1$ for two agents] 
	If $n=2$, the BH rule $\varphi ^{1}$ for bads is the
mirror image of the dominant TH rule $\varphi ^{1}$ (\ref{23}): 
\begin{equation*}
\varphi ^{1}(x)=\left\{\begin{array}{cc}
(1,0), & \ \frac{x_{1}}{x_2}\leq \frac{1}{2}\\
\left(\frac{%
	x_{2}}{2x_{1}},\ 1-\frac{x_{2}}{2x_{1}}\right),  & \frac{1}{2}\leq
\frac{x_1}{x_2}\leq 1 
\end{array}\right.. 
\end{equation*}
\begin{figure}[h!]
	\centering
	\vskip -0.5 cm
	{\includegraphics[width=6cm, clip=true]{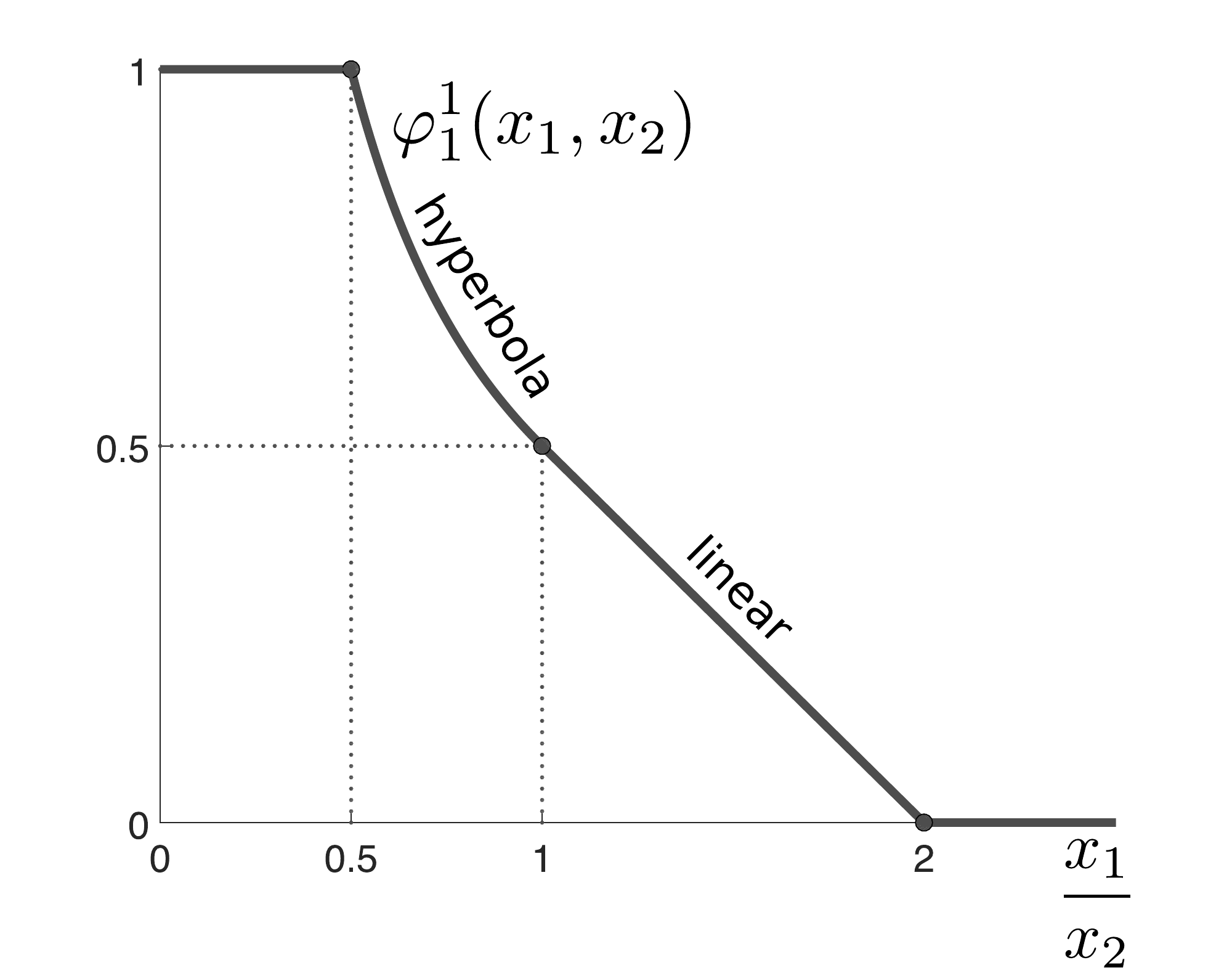}}
	\caption{The share of the first agent under the BH rule $\varphi^1$ for two agents as a function of $\frac{x_1}{x_2}$.}
	\label{fig3}
\end{figure}
\end{example}

\begin{theorem}[for bads]\label{th2}
\textit{For any }$n\geq 2$\textit{, the Bottom-Heavy rule }$%
\varphi ^{1}$\textit{ dominates every  other  {symmetric} {API} rule for
	bads satisfying Fair Share.}
\end{theorem}

{In Appendix~\ref{appA}  we define a family of BH rules $\varphi^{\theta}, \ \theta\in[0,1]$ and first show, as in the case of goods, that any other rule is dominated by some $\varphi^\theta$; then we check that $\varphi^1$ dominates $\varphi^\theta$ for $\theta<1$.}  {This additional domination argument within the	family of BH rules $\varphi ^{\theta }$ is straightforward but
	lengthy.}

\section{Worst-case performances}\label{sec5}

\textit{Notation. }  We write $\Phi$ for the set of symmetric almost prior-dependent rules $\varphi$, $\Phi(FS)$
	for  rules $\varphi\in \Phi$ satisfying Fair Share, and $\Phi_{ind}(FS)$ for symmetric rules  
	 $\varphi\in \Phi(FS)$. Thus $\Phi_{ind}(FS)\subset \Phi(FS)\subset\Phi$. 
	{Let ${\Pi}_{n}$ be the set of
	normalized problems with $n$ agents.} {Finally we recall that $S(\varphi,\mathcal{P})$ denotes the expected social welfare (expected social cost); see~\eqref{eq_SW}.}
\begin{definition}\label{def10}
	{{The  \textit{Competitive Ratio}\footnote{ { {The term ``competitive ratio'' is borrowed from the literature on online algorithms:} there it is defined as a worst-case {factor by which the value of the objective  (the social welfare in our model) for an online rule is less than the value achieved by the best offline rule,} where the manager has full knowledge of the future.\\
	 \indent{Our model can be interpreted as online allocation problem with i.i.d. objects; see Section~\ref{subsect_lit}. Under this interpretation our definition of competitive ratio matches the traditional one. Knowing the future reduces to knowing} the empirical distribution of the future sequence of values, which in an i.i.d. environment with a large number of repetitions converges to the prior. Thus the best offline rule becomes just the best prior-dependent rule {in the long run}.}}
				 \textit{(CR)} of an {API} rule $\varphi\in \Phi_{ind}(FS)$ is defined as follows:}}
	\[
	\mbox{for a good: } \ CR_{n}(\varphi )=\sup_{\mathcal{P}\in 
		{\Pi }_{n}}\sup_{\psi \in \Phi(FS)}\frac{S (\psi ,\mathcal{P})}{S (\varphi ,\mathcal{P})} \qquad\quad
	\mbox{for a bad: } \ CR_{n}(\varphi )=\sup_{\mathcal{P}\in 
		{\Large \Pi }_{n}}\sup_{\psi \in \Phi(FS)}\frac{S (\varphi ,\mathcal{P})}{S (\psi ,\mathcal{P})}. 
	\]%
	{The CR identifies the worst-case loss in the social welfare caused by {almost} prior-independence.}

	For a good and a rule $\varphi \in
	\Phi (FS)$, we write $\pi(\varphi,\mathcal{P})$ for the ratio
of the optimal unconstrained {social welfare}  generated by the Utilitarian rule to the {social welfare} generated by $\varphi$. {For
			a bad, it is the ratio  of the {social cost}
			generated by }$\varphi ${ to the optimal {social cost}:}{%
			{\[
				\mbox{for a good: } \pi(\varphi,\mathcal{P})=\frac{\mathbb{E}%
					_{\mu }\left(\max_{i}X_{i}\right)}{S (\varphi ,\mathcal{P})}\ \ \mbox{for a bad: }\ \pi(\varphi ,\mathcal{P})=\frac{S (\varphi ,\mathcal{P})}{\mathbb{E}%
					_{\mu }\left(\min_{i}X_{i}\right)}. 
				\]%
			}
		}
		
		{The \textit{Price of Fairness (PoF)} of} $\varphi
		\in \Phi (FS)$ {is the worst possible ratio $\pi(\varphi,\mathcal{P})$:}
		 {\[
			PoF_{n}(\varphi )={\sup_{\mathcal{P}\in \Pi_{n}}\pi(\varphi ,\mathcal{P})}%
			\geq 1. 
			\]}%
	
\end{definition}	
\begin{lemma}\label{lm1}
{If the {API} rule }$\varphi \in \Phi_{ind}
	(FS)$\textit{ divides  {a good}, we have}%
	\[
	CR_{n}(\varphi )=PoF_{n}(\varphi )=\sup_{x\in 
		\mathbb{R}
		_{+}^{N}}\frac{\max_{i}x_{i}}{\sum_{i\in N}\varphi _{i}(x)\cdot x_{i}}. 
	\]%
	\textit{If }$\varphi \in \Phi_{ind}(FS)$\textit{ divides  {a bad}, we have}%
	\[
	CR_{n}(\varphi )=PoF_{n}(\varphi )=\sup_{x\in 
		\mathbb{R}
		_{+}^{N}}\frac{\sum_{i\in N}\varphi _{i}(x)\cdot x_{i}}{\min_{i}x_{i}}. 
	\]
\end{lemma}
\begin{proposition}[for goods]\label{prop4} \
	\begin{enumerate}
	\item \textit{The }$CR_{n}$\textit{ of any rule }$\varphi \in
	\Phi_{ind} (FS)$\textit{ is at most }$n$\textit{; the $CR_n$ of Equal Split is
		exactly }$n$\textit{.}
	
	\item \textit{The} $CR_{n}$\textit{ of the Proportional rule is 
	} {$\frac{\sqrt{n}}{2}+\frac{1}{2}$}\textit{. For instance, }$121\%$\textit{ for }$n=2$.
	
	\item \textit{The }$CR_{n}$\textit{ of the Top-Heavy rule }$%
	\varphi ^{\theta }$\textit{ is decreasing in }$\theta $. \textit{Moreover:}%
	\[
	CR_{n}(\varphi ^{1})=\frac{n}{2\sqrt{n}-1}=\frac{\sqrt{n}}{2}+\frac{1}{4}+O\left(%
	\frac{1}{\sqrt{n}}\right), 
	\]%
	\[
	CR_{n}(\varphi ^{\theta })=\frac{n}{2\sqrt{(n-1+\theta )\theta }+1-2\theta }%
	\geq CR_{n}(\varphi ^{1}). 
	\]
	
 \textit{For instance, }$CR_{2}(\varphi ^{1})\simeq 109\%$\textit{
		for }$n=2$\textit{.}
	
	\item \textit{The smallest  }$PoF_{n}$\textit{ of a {
			prior-dependent} rule in }$\Phi (FS)$\textit{ is such that}%
	\[
	\frac{n}{2\sqrt{n}-1}\geq \inf_{\varphi\in \Phi(FS)}PoF_{n}(\varphi )\geq \frac{n}{2%
		\sqrt{n}-\frac{1}{2}}=\frac{\sqrt{n}}{2}+\frac{1}{8}+O\left(\frac{1}{\sqrt{n}}\right). 
	\]
	
 \textit{For }$n=2$\textit{ it is }$108\%$\textit{.\smallskip }
	\end{enumerate}
	\end{proposition}
Statements $iii)$ and $iv)$, together with Lemma~\ref{lm1}, make clear that the $PoF_{n}$ of the TH rule $\varphi ^{1}$\ is essentially the best  {$PoF%
	_{n}$} of
any fair prior-dependent rule.
\begin{proposition}[for bads]\label{prop5} \
	\begin{enumerate}
	\item \textit{The }$CR_{n}$\textit{ of Equal Split is unbounded
	(for any fixed }$n$\textit{) and that of the Proportional rule is }$n$.

\item \textit{ The CR}$_{n}$\textit{of the Bottom-Heavy rule }$\varphi ^{1}$\textit{ is such that}%
$$
\frac{n}{4}+\frac{5}{4}\geq CR_{n}(\varphi ^{1})\geq \frac{n}{4}+\frac{1}{2}%
+\frac{1}{4n}. 
$$
\textit{It is }$109\%$\textit{ for }$n=2$\textit{.}

\item \textit{ The smallest} $PoF_{n}$\textit{of a
	prior-dependent rule in }$\Phi(FS)$\textit{ is}%
$$
\inf_{\varphi \in \Phi(FS)}PoF_{n}(\varphi )=\frac{n}{4}+\frac{1}{2}+\frac{1}{4n}. 
$$
\textit{For }$n=2$\textit{ it is }$108\%$\textit{.}
\end{enumerate}
\end{proposition}

\medskip
Again, the last two statements and Lemma~\ref{lm1} imply that the PoF$_{n}$ of the
BH rule $\varphi ^{1}$\ is essentially the best PoF$_{n}$ of any fair {prior-dependent} rule.

All three results (Lemma~\ref{lm1} and Propositions~\ref{prop5} and \ref{prop6}) are proved in Appendix~\ref{appB}.

\section{Asymptotic performance for standard distributions}\label{sec6}

We evaluate the performance of the TH,  BH, and 
Proportional rules in the benchmark setting where the number of agents is
large and their values are given by independent and identically distributed
(i.i.d.) random variables. 

 {We will see that in this setting the TH rules behave significantly better than under the worst-case assumption of Section~\ref{sec5}. In fact they keep a constant fraction of {the optimal social welfare} even for a large number of agents. The conclusion is almost the same for the BH rule, except for a certain subclass of distributions with support touching zero, for which the social cost can exceed the optimal social cost by a factor  of $O(\sqrt{n})$ (still much better than $O(n)$ in the worst case). The Proportional rule does much worse in several natural i.i.d. contexts detailed 
below.}

Fix a distribution ${\nu }\in \Delta (\mathbb{R}_{+})$ with unit mean and
assume that the vector $X=(X_{i})_{i=1,\ldots,n}$ of values is distributed
according to $\mu =\otimes _{i=1}^{n}\nu $; i.e., the values are independent
random variables with distribution $\nu $. The corresponding problem $%
\mathcal{P}_{n}(\nu )$ is normalized.

In Appendix~\ref{appC} we derive the somewhat cumbersome general formulas describing
the  {ratio} $\pi(\varphi ,\mathcal{P}_{n}(\nu ))$ when $n$ is large. Here we discuss examples and corollaries of
the general results.

\subsection{A good}

\subsubsection{Bounded support: $\protect\nu $ is the uniform distribution
	on $[0,1]$.}

In this case the TH rule $\varphi ^{1}$ and the Proportional rule $\varphi
^{pro}$ have similar  performances.

For $n=2$, the TH almost achieves the optimal welfare level. The Proportional
rule is $10\%$ behind: simple computations show that $\pi(\varphi ^{1},%
\mathcal{P}_{2}(\mathrm{uni}[0,1]))=\frac{8}{5+4\ln 2}\approx 1.03$ and $\pi(\varphi ^{pro},%
\mathcal{P}_{2}(\mathrm{uni}[0,1]))=\frac{2}{\ln2 -1}\approx 1.13$. Compare these numbers with
the worst-case guarantees from Proposition~\ref{prop4}: ${PoF_{2}(\varphi
	^{pro})}=\frac{\sqrt{2}+1}{2}\approx 1.21$ and ${PoF_{2}(\varphi
	^{1})}=\frac{2}{2\sqrt{2}-1}\approx 1.09$. We see that the Proportional rule
generates less social welfare for the uniform distribution than the TH rule for 
\textit{any} distribution.

For $n\rightarrow \infty $, Proposition~\ref{prop6} from Appendix~\ref{appC} and Lemma~\ref{lm3} below
imply that the  ratios for our two rules converge  {and the limit values are}
\[
\pi(\varphi ^{1},\mathcal{P}_{\infty }(\mathrm{uni}[0,1]))=\frac{1}{\frac{1}{16}+\ln
	2}\approx 1.32\ \ \mbox{and}\  \ \pi(\varphi ^{pro},\mathcal{P}_{\infty }(\mathrm{uni%
}[0,1]))= 1.5. 
\]

This result is in  sharp contrast with the worst-case behavior  {(Section~\ref{sec5}):}
there are problems $\mathcal{P}$ with $n$ agents such that the TH rule
generates only a $2/\sqrt{n}$ fraction of {the optimal social welfare.} Our next
result generalizes this observation.

\subsubsection{The TH rule keeps a positive fraction of the {optimal social welfare.}}

This holds in general, not just in the above example. Fix a distribution $%
\nu $ with mean $1$ and with non-zero average absolute deviation $D(\nu )=\int
|x-1|d\nu (x)$. Note that $D(\nu)$ is at most $2$. \smallskip

\begin{lemma}\label{lm2}
{If $\nu $ has mean $1$ and a finite moment  {$\int_{%
				\mathbb{R}_{+}}x^{\beta }d\nu (x)$} for some $\beta >2$, then the
		 ratio for the TH rule converges to a limit value that
		satisfies the following upper bound:
		\begin{equation}
		\pi(\varphi ^{1},\mathcal{P}_{\infty }(\nu))\leq \frac{2}{D}+\frac{4}{D^2}.
		\label{lowerbound}
		\end{equation}%
		If in addition $\nu $ has unbounded support, then}%
	\begin{equation}
	\pi(\varphi ^{1},\mathcal{P}_{\infty }(\nu))\geq \frac{1}{D}.  \label{upper bound}
	\end{equation}%
\end{lemma}	
{The proof is in Appendix~\ref{appC}.} For instance, if $\nu $ is the exponential distribution we have%
\[
\pi(\varphi ^{1},\mathcal{P}_{\infty }(\mathrm{exp}))=\frac{1}{1-2e^{-\frac{1}{2}}-%
	\mathrm{Ei}(-1/2)}\approx 2.88, 
\]%
where $\mathrm{Ei}$ stands for a special function, the exponential integral.\footnote{$\mathrm{Ei}(x)=-\int_{-x}^\infty\frac{e^{-t}}{t}\,dt.$}
Contrast this with the situation for the Proportional rule.
\begin{lemma}\label{lm3}
{Under the assumptions of Lemma~\ref{lm2},} 
\[
\pi(\varphi ^{{pro}},\mathcal{P}_{n}(\nu ))=\frac{\mathbb{E}_{\mu }\left(\max_{i}X_{i}\right)}{\mathbb{E}_{\nu }(X_{1})^{2}}(1+o(1)),\quad \mbox{ as  }n\rightarrow
\infty, 
\]
\textit{(where }$a_{n}=o(1)$\textit{ means that }$a_{n}\rightarrow 0$%
\textit{, as }$n\rightarrow \infty $\textit{).}
\end{lemma}

Indeed, by the law of large numbers,
\[
S (\varphi ^{{pro}},\mathcal{P}_{n}(\nu ))=\mathbb{E}_{\mu }\left(\sum_{i\in
	N}X_{i}\varphi _{i}^{{pro}}(X)\right)=n\cdot \mathbb{E}_{\mu }\left(\frac{(X_{1})^{2}}{%
	\sum_{i\in N}X_{i}}\right)\longrightarrow \mathbb{E}_{\mu }\left(\frac{(X_{1})^{2}}{\mathbb{E}%
	_{\mu }X_{1}}\right)=\mathbb{E}_{\nu }(X_{1})^{2}. 
\]

Lemma~\ref{lm3} implies that $\pi(\varphi ^{{pro}},\mathcal{P}_{\infty }(\nu ))$ tends
to $+\infty$ if $\nu $ has unbounded support, because $\mathbb{E}_{\mu
}\max_{i}X_{i}$ tends to infinity. For instance, $\pi(\varphi ^{pro},\mathcal{P}%
_{n}(\mathrm{exp}))=\frac{\ln n}{2}(1+o(1))$.

Of course, this limit is finite if the support of $\nu $ is bounded.

\subsection{A bad}

When a bad is divided, the performance of the BH and Proportional rules
is determined by the behavior of the distribution at the left-most point of
the support. Both rules generate a bounded multiple of {the optimal social cost}
 when $0$ does not belong to the support of $\nu $; the BH rule also does
 well when $\nu $ has a non-zero density at $0$. However, both rules have
poor performance if the support touches $0$ but $\nu $ has not enough
``weight'' near $0$. Here we give three examples to illustrate the general
asymptotic results of Appendix~\ref{appC}.

\subsubsection{The support does not touch zero: $\protect\nu $ is uniform on 
	$[\frac{1}{2},\frac{3}{2}]$.}

By Proposition~\ref{prop7} in Appendix~\ref{appC}, the ratios for the BH and
 Proportional rules converge to limit values that are pretty close to
each other:%
\[
\pi\left(\varphi^{1},\mathcal{P}_{\infty }\left(\mathrm{uni}\left[\frac{1}{2},\frac{3}{2}\right]\right)\right)=%
{e-1}\approx 1.72\ \ \mbox{and}\ \ \pi\left(\varphi ^{pro},\mathcal{P}_{\infty
}\left(\mathrm{uni}\left[\frac{1}{2},\frac{3}{2}\right]\right)\right)=\frac{2}{\ln 3}\approx 1.82. 
\]

\subsubsection{The support touches zero but there is not enough weight
	around it: $\protect\nu $ has density $\frac{3}{4}x(2-x)$ on $[0,2]$.}

For this distribution, {the optimal social cost} tends to zero while the
losses of the BH and  Proportional rules remain positive. Proposition~\ref{prop7}
shows that the  ratios for both rules tend to infinity at the
speed of $\sqrt{n}$, while the ratio for the BH rule remains  $\frac{1}{\sqrt{3}}%
\approx 0.58$ times lower than that for the Proportional one:%
\[
\pi(\varphi ^{1},\mathcal{P}_{n}(\nu ))=\frac{2}{3\sqrt{\pi }}{\sqrt{n}%
}(1+o(1))= \pi(\varphi ^{pro},\mathcal{P}_{n}(\nu ))\frac{1}{\sqrt{3}}(1+o(1)). 
\]

\subsubsection{The distribution has non-zero density at $0$ (e.g., $\protect%
	\nu $ is uniform on $[0,2]$).}

In this case, the BH rule outperforms the Proportional one in the limit.

\begin{lemma}\label{lm4}
{Assume that the distribution $\nu $ has a
	continuous density $f$ on an interval $[0,a]$ and $f(0)>0$. Then $\pi(\varphi
	^{1},\mathcal{P}_{n}(\nu ))$ converges to a finite limit as $n$ becomes
	large, whereas $\pi(\varphi ^{pro},\mathcal{P}_{n}(\nu ))=\Omega\left( \frac{n}{%
		\ln (n)}\right) $ as\footnote{%
		Recall that $a_{n}=\Omega(b_{n})$ if there exist $n_{0}$ and $C>0$ such that $%
		|a_{n}|\geq C|b_{n}|$ for all $n\geq n_{0}$.} $n\rightarrow \infty $.}
\end{lemma}

A similar result for the case where the density is infinite at $x=0$ is the
subject of Lemma~\ref{lm5} in Appendix~\ref{appC}.

The statement about the BH rule follows from the asymptotic result for the
order statistic: the expected values of {$X^{(k)}$} for small numbers $k$
are equal to $\frac{k}{f(0)\cdot n}(1+o(1))$ as\footnote{%
	The order statistic $X^{(k)}$ has the same distribution as {$F^{-1}(Y^{(k)})$,}
	where $F$ is the distribution function of $\nu $ and $Y_{i}$, $i\in N$, are
	independent random variables uniformly distributed on $[0,1]$. By symmetry, $\mathbb{%
		E}(Y^{(k)})=\frac{k}{n+1}$.} $n\rightarrow \infty $. %
 Therefore, on average, only a bounded number of
agents with the smallest $X_{i}$ receive a non-zero portion of a bad, which
implies that the  ratio is bounded away from infinity.

For the Proportional rule, we have $S(\varphi^{pro},\mathcal{P}%
_n(\nu))=n\cdot\mathbb{E}\left(\frac{1}{\sum_k \frac{1}{X^{(k)}}}\right)$. For large $n$,
we can estimate the denominator from below by the harmonic series; taking
into account that $\mathbb{E} (X^{(1)})=\frac{1}{f(0)\cdot n}(1+o(1))$ we get the
desired asymptotic formula.

\section{Extensions}\label{sec7}
\subsection*{Envy-Freeness}
An alternative, much more demanding interpretation of fairness in our model
is (ex~ante) \textit{Envy-Freeness}, which means, in the case of a good:%
\[
\mathbb{E}_{\mu }(\varphi _{i}^\mu(X^*)\cdot X_{i}^*)\geq \mathbb{E}_{\mu
}(\varphi _{j}^\mu(X^*)\cdot X_{i}^*) \ \mbox{ for all } \ i,j \ \mbox{ and } \ \mathcal{P}=(N,\mu,X).
\]
Fixing $i$ and summing up the $n$ inequalities given above  when $j$ covers $N$ 
(including $j=i$), we see that Envy-Freeness implies Fair Share.

The critical Proposition~\ref{prop2} can be adapted as follows. Set $g(x)=(\varphi
_{1}(x)-\varphi _{2}(x))\cdot x_{1}$ so that Envy-Freeness  {for a symmetric {API} rule is equivalent to} $\mathbb{E}%
_{\mu }(g(X))\geq g(e^{N})=0$ whenever $\mathbb{E}_{\mu }(X)=e^{N}$, and
deduce in the same way that there is a vector $\beta \in 
\mathbb{R}
^{n}$ such that $(\varphi _{1}(x)-\varphi _{2}(x))\cdot x_{1}\geq \beta
\cdot (x-e^{N})$ for all $x$. By the symmetry of $\varphi $ and $\varphi (x)\in
\Delta (N)$, it is immediate that there exists $\theta \geq 0$ such that, for
any $x$ with weakly increasing coordinates,
\[
\theta \left(1-\frac{x_{i-1}}{x_{i}}\right)\leq \varphi _{i}(x)-\varphi _{i-1}(x)\leq
\theta \left(\frac{x_{i}}{x_{i-1}}-1\right)\mbox{ for all }i=1,\ldots,n. 
\]%
Applying this when $x_{i}$ is a geometric sequence with a large exponent
gives $\theta \leq \frac{2}{n(n-1)}$, and by choosing $\theta ^{\ast }=\frac{%
	2}{n(n-1)}$ and defining $\varphi $ appropriately, we guarantee {the PoF}
 {of the order of $n$, comparable} to the  {minimal}
Price of Envy-Freeness for  prior-dependent rules; see \cite{Caragianis2009}.

{For a bad, Envy-Freeness is defined with the opposite sign in the inequality.} Similarly, we have that if the coordinates of $x$ are weakly
increasing, an Envy-Free rule $\varphi $ is such that%
\[
\theta \left(1-\frac{x_{i-1}}{x_{i}}\right)\leq \varphi _{i-1}(x)-\varphi _{i}(x)\leq
\theta \left(\frac{x_{i}}{x_{i-1}}-1\right)\mbox{ for all }i=1,\ldots,n, 
\]%
where again the parameter $\theta $ is at most $\frac{2}{n(n-1)}$. However, this
time the  performance of such a rule is fairly poor, as one can
see with $\theta ^{\ast }=\frac{2}{n(n-1)}$ and the disutility profile $%
x_{i}=2^{i-1}$ for all $i$. The most efficient profile of  {shares} is then $%
\varphi _{i}(x)=(n-i)\theta ^{\ast }$ and the ratio $\frac{1}{x_{1}}%
(\sum_{1}^{n}\varphi _{i}(x)x_{i})$ is then in the order of $\frac{2^{n}}{%
	n^{2}}$.
\subsection*{Asymmetric  {ownership rights}}
If the agents are endowed with unequal ownership rights on the object,
captured by the shares $\lambda \in \Delta (N)$, it is natural to adapt Fair
Share as follows (for goods): {$\mathbb{E}_{\mu }(\varphi _{i}(X^*)\cdot
X_{i}^*)\geq \lambda _{i}$} for all $i$. We can again
adapt the argument in Proposition~\ref{prop2} to characterize this constraint by the
existence, for each $i$, of a linear form lower-bounding the function $%
x\rightarrow \varphi _{i}(x)\cdot x_{i}$. But  {we cannot use arguments based on symmetry to reduce the number of free parameters} and the characterization of the undominated fair
rules is much more difficult.
 { \subsection*{Mixture of goods and bads}
	The case of objects with utility of a random sign is interesting but
	difficult: we cannot apply our technique when expected utilities can be
	zero; even if this case is ruled out, when realized utilities are both
	positive and negative, the rule must for efficiency divide the object
	between positive utility agents only and this random change in the size of
	the recipients throws off our computations, starting with the Proportional
	rule (Proposition~\ref{prop1}) and the key Propositions~\ref{prop2} and~\ref{prop3}.}

\section{Conclusion}\label{sec8}

We initiate the discussion of fair division problems, where the manager has limited access to statistical information about  the realized utilities (or disutilities). Such limitations are the major concerns in literatures on robust mechanism design and online algorithms, but as far as we know, they have never been discussed in the field of fair division. 
	
We discuss a prototypical fair division problem with just one random object to divide.  
The setting proved to be quite rich and at the same time tractable enough for the explicit description of the best rules: the entirely new families of the Top-Heavy and Bottom-Heavy rules. By contrast, the literatures on robust mechanism design and online algorithms typically find a certain approximation to the best rules, and, in the rare cases where the best rules are described, the best rules turn out to be previously known ones. 

Having the best rules in hand allowed us to push the analysis further and compute the exact values of the Competitive Ratio and the Price of Fairness. Then we found that, surprisingly, the risk-averse manager knowing the first moments of the underlying distribution can do almost as well as the manager having detailed statistical information. We do not have an intuitive explanation of this effect. Understanding it in greater depth and describing environments that exhibit similar phenomena would be a challenging avenue for future research.    

The paper suggests many concrete theoretical open questions, e.g., the extension of the results to the setting with many random objects delivered at once, and the other questions touched on in Section~\ref{sec7}.

\bibliographystyle{nonumber}

\appendix

\section{Proofs for Section~\ref{sec4}}\label{appA}

\subsection{Proof of Proposition~\ref{prop3}}

\begin{proof}[{The }``\textbf{if}'' statement. ]  The proof is the same as the ``if'' statement in
Proposition~\ref{prop2} for goods, after reversing the inequalities.

\smallskip
\noindent\textit{{The }``\textbf{only if}'' statement. } The proof is similar to the ``only if'' statement in Proposition~\ref{prop2}. Fix an  {API} rule $%
\varphi $ satisfying FS and define $f(x)=\varphi _{1}(x)\cdot x_{1}$; by
symmetry, $f(e^{N})=\frac{1}{n}$. For any coefficients $\mu \in \Delta
(K)$ and convex combination $\sum_{k=1}^{K}\mu _{k}y^{k}=e^{N}$ in $%
\mathbb{R}
_{+}^{N}$, we apply FS to the normalized problem in which $X=y^{k}$ with
probability $\mu _{k}$ and obtain $\sum_{k=1}^{K}\mu _{k}f(y^{k})\leq
f(e^{N})$. Therefore the concavification $g$ of $f$ coincides with $f$ at $%
e^{N}$, and there is some $\alpha \in 
\mathbb{R}
^{N}$ supporting its graph at $(e^{N},g(e^{N}))$, i.e.,%
\[
\varphi _{1}(x)\cdot x_{1}\leq \alpha \cdot (x-e^{N})+\frac{1}{n}%
\mbox{ for
	all }x\in 
\mathbb{R}
_{+}^{N}. 
\]%
The same symmetry arguments show that $\alpha $\ takes the form $\alpha
=(\alpha _{1},\beta ,\beta ,\ldots ,\beta )$ and $\alpha \cdot e^{N}=\frac{1%
}{n}$. This time the inequality $0\leq \varphi _{1}(x)\cdot x_{1}\leq \alpha
_{1}x_{1}+\beta x_{N\diagdown \{1\}}$ implies $\alpha \geq 0$. Setting $\delta
=n\beta $ and rearranging we get:%
\[
\varphi _{i}(x)\leq \frac{1}{n}+\delta \left(\frac{\overline{x}}{x_{i}}-1\right)%
\mbox{
	for all }i\in N\mbox{ and }x\in 
\mathbb{R}
_{+}^{N}.
\]

 {It remains to find the bound on $\delta$.} Because $\frac{\overline{x}}{x_{i}}\geq \frac{1}{n}$, the inequality $\varphi
_{i}(x)\geq 0$ holds everywhere if it holds at $x=e^{\{i\}}$, where it
implies the bound $\delta \leq \frac{1}{n-1}$. Then the change of parameters 
$\theta =(n-1)\delta $ implies the desired inequality~(\ref{22}). 
\end{proof}

\subsection{Proof of Theorem~\ref{th2}}
\begin{proof}[Step 1.]
First we define the whole family of Bottom-Heavy rules for $\theta\in[0,1]$ by
	\begin{equation}
	\varphi _{i}^{\theta }(x)=\left\{
	\begin{array}{cc}
	{\displaystyle    \frac{1}{n}+\frac{\theta }{n-1}\left(\frac{\overline{x}%
		}{x_{i}}-1\right),} &   i: \ x_{i}\leq x^{(\widetilde{t})} \\
	{\displaystyle \frac{1}{|\sigma (x;\widetilde{t}+1)|}\left(%
		1-\sum_{i:x_{i}\leq x^{(\widetilde{t})}}\varphi _{i}^{\theta }(x)\right),} & i\in\sigma (x;\widetilde{t}+1)\\
	0, & \mbox{otherwise}
	\end{array}
	\right.,
	\end{equation}
	{where $\widetilde{t}$ is the maximal $t=0,1,2,\dots,n$ such that $\sum_{i:\,x_{i}\leq x^{(t)}}\large \left(\frac{1}{n}+\frac{\theta 
		}{n-1}\left(\frac{\overline{x}}{x_{i}}-1\right)\right)\leq 1$.} The definition is correct since $ \frac{1}{n}+\frac{\theta }{n-1}\left(\frac{\overline{x}%
	}{x_{i}}-1\right)$ is always non-negative and $\sum_{i\in N} \frac{1}{n}+\frac{\theta }{n-1}\left(\frac{\overline{x}%
	}{x_{i}}-1\right)\geq 1$ with strict inequality for $x$ not parallel to $e^N$ and $\theta\ne 0$.

\smallskip
\noindent\textit{Step 2.} Next we prove that if the  {API}
rule $\varphi $ satisfies inequalities (\ref{22}) for some $\theta,\ 0\leq \theta
\leq 1$, then $\varphi ^{\theta }$ dominates $\varphi $ or equals $\varphi $%
. In Step 3 we show that $\varphi ^{1}$ dominates $\varphi ^{\theta }$ if $%
\theta <1$.

First, for $\theta =0$, inequalities (\ref{22}) imply that $\varphi $ itself
is Equal Split, i.e., $\varphi ^{0}$. From now on, we assume that $\theta >0$.

Along the ray through $e^{N}$ the rules $\varphi $ and $\varphi ^{\theta }$
coincide by symmetry. Now we fix $x\in 
\mathbb{R}
_{+}^{N}$ not parallel to $e^{N}$ and let $\widetilde{t}$ be defined as above. From (\ref{22}) we get $\varphi _{i}(x)\leq \varphi _{i}^{\theta }(x)$
for all $i$ such that $x_{i}\leq x^{(\widetilde{t})}$. Thus%
\begin{equation}
\sum_{i:\,x_{i}\leq x^{(\widetilde{t})}}(\varphi _{i}(x)-\varphi
_{i}^{\theta }(x))x_{i}\geq \sum_{i:\,x_{i}\leq x^{(\widetilde{t})%
}}(\varphi _{i}(x)-\varphi _{i}^{\theta }(x))x^{(\widetilde{t}+1)}.
\label{19}
\end{equation}%
Next we have $\sum_{i:\,x_{i}\geq x^{(\widetilde{t}+1)}}\varphi
_{i}^{\theta }(x)x_{i}=\sum_{i:\,x_{i}\geq x^{(\widetilde{t}+1)}}\varphi
_{i}^{\theta }(x)x^{(\widetilde{t}+1)}$ because $\varphi _{i}^{\theta
}(x)=0$ if $x_{i}>x^{(\widetilde{t}+1)}$. Thus%
\begin{equation}
\sum_{i:\,x_{i}\geq x^{(\widetilde{t}+1)}}(\varphi _{i}(x)-\varphi
_{i}^{\theta }(x))x_{i}\geq \sum_{i:\,x_{i}\geq x^{(\widetilde{t}%
		+1)}}(\varphi _{i}(x)-\varphi _{i}^{\theta }(x))x^{(\widetilde{t}+1)}.
\label{20}
\end{equation}%
Summing up these two inequalities gives the corresponding weak inequality (%
\ref{17}).

Assume finally that all inequalities (\ref{17}) are equalities. If at least
one $x_{i}$ is zero, (\ref{22}) implies that $\varphi (x)$ does not put any
weight outside $\sigma (x,1)$, and so $\varphi (x)=\varphi ^{\theta }(x)$. If
each $x_{i}$ is strictly positive, our assumption implies that (\ref{19}) is
an equality; but the definition of $\widetilde{t}$ implies $x^{(\widetilde{t})}<x^{(\widetilde{t}+1)}$. Therefore $\varphi _{i}(x)=\varphi _{i}^{\theta }(x)$ as
long as $x_{i}\leq x^{(\widetilde{t})}$. Now (\ref{20}) cannot be an
equality if $\varphi (x)$ puts any weight on agents with disutilities greater
than $x^{(\widetilde{t}+1)}$, and we conclude that $\varphi (x)=\varphi
^{\theta }(x)$ by the symmetry of $\varphi $.\smallskip

\smallskip
\noindent\textit{Step 3.}We show that $\varphi ^{\theta ^{+}}$ dominates $%
\varphi ^{\theta ^{-}}$ if $\theta ^{+}>\theta ^{-}>0$. We write these rules as
$\varphi ^{+}$ and $\varphi ^{-}$ for simplicity, and fix $x\in 
\mathbb{R}
_{+}^{N}$. For $\varepsilon =+,-$, denote $\widetilde{t}$ for $\varphi ^{\varepsilon }(x)$ by 
$t^{\varepsilon }$.

We use the notation%
\[
\delta _{i}=\frac{1}{n-1}\left(\frac{\overline{x}}{x_{i}}-1\right)\mbox{ and }\psi
_{i}^{\varepsilon }=\frac{1}{n}+\theta ^{\varepsilon }\delta _{i}.
\]

We prove inequality (\ref{17}) between $\varphi ^{+}$ and $\varphi ^{-}$ for
a vector $x$ with no two equal coordinates. This will be enough because each
mapping $\varphi ^{\theta }$ is only discontinuous at $x$ if $|\sigma (x,%
\widetilde{t}+1)|>1$, and the total disutility $\sum_{i\in N}\varphi _{i}^{\theta
}(x)x_{i}$ is continuous at such points.

Finally we label the coordinates of $x$ increasingly, so that $x_{i}=x^{(i)}$ for all $i$, and the definition of $\varphi ^{\varepsilon }(x)$ is
notationally simpler: $\varphi _{i}^{\varepsilon }(x)=\psi _{i}^{\varepsilon
}>0$ for $1\leq i\leq t^{\varepsilon }$; $0\leq \varphi _{t^{\varepsilon
	}+1}^{\varepsilon }(x)<\psi _{t^{\varepsilon }+1}^{\varepsilon }$ ; $\varphi
_{j}^{\varepsilon }(x)=0$ for $j>t^{\varepsilon }+1$.

We claim first that  $t^{+}\leq t^{-}$, and if $t^{+}=t^{-}=t$ then $\lambda =%
\frac{\varphi _{t^{+}+1}^{+}(x)}{\psi _{t^{+}+1}^{+}}<\mu =\frac{\varphi
	_{t^{-}+1}^{-}(x)}{\psi _{t^{-}+1}^{-}}$, where $0\leq \lambda ,\mu <1$. To
prove this we compute%
\[
1=\sum_{1}^{t^{+}}\psi _{i}^{+}+\lambda \psi _{t^{+}+1}^{+}=\frac{%
	t^{+}+\lambda }{n}+\theta ^{+}(\delta _{\{1,\ldots ,t^{+}\}}+\lambda \delta
_{t^{+}+1}). 
\]%
As $\frac{t^{+}+\lambda }{n}<1$, this implies $\delta _{\{1,\ldots
	,t^{+}\}}+\lambda \delta _{t^{+}+1}>0$; therefore,%
\[
1>\frac{t^{+}+\lambda }{n}+\theta ^{-}(\delta _{\{1,\ldots ,t^{+}\}}+\lambda
\delta _{t^{+}+1}). 
\]%
However, by repeating the computation above for $\varphi ^{-}(x)$ we get%
\[
1=\frac{t^{-}+\mu }{n}+\theta ^{-}(\delta _{\{1,\ldots ,t^{-}\}}+\mu \delta
_{t^{-}+1}). 
\]%
We see that $t^{-}<t^{+}$ leads to a contradiction between the last two
statements. Also, if $t^{-}=t^{+}=t$ they imply $\lambda \psi _{t+1}^{-}<\mu
\psi _{t+1}^{-}$, and  so $\lambda <\mu $ because $\psi _{i}^{-}>0$ for all $i$.
The claim is proved.

Next we evaluate the difference $\Delta $ in total disutility generated by
our two rules:%
\[
\Delta =\sum_{N}{\large (}\varphi _{i}^{+}(x)-\varphi _{i}^{-}(x){\large )}%
x_{i}= 
\]%
\[
=\sum_{1}^{t^{+}}(\psi _{i}^{+}-\psi _{i}^{-})x_{i}+(\lambda \psi
_{t^{+}+1}^{+}-\psi _{t^{+}+1}^{-})x_{t^{+}+1}-\sum_{t^{+}+2}^{t^{-}}\psi
_{i}^{-}x_{i}-\mu \psi _{t^{-}+1}^{-}x_{t^{-}+1}, 
\]%
where we have assumed that $t^{+}<t^{-}$; if instead $t^{+}=t^{-}=t$ the last
three terms of the sum reduce to $(\lambda \psi _{t+1}^{+}-\mu \psi
_{t+1}^{-})x_{t^{+}+1}$. As $x_{i}$ is increasing in $i$ we have%
\[
\Delta \leq \sum_{1}^{t^{+}}(\psi _{i}^{+}-\psi _{i}^{-})x_{i}+\lambda \psi
_{t^{+}+1}^{+}x_{t^{+}+1}-(\psi _{\{t^{+}+1,\ldots ,t^{-}\}}^{-}+\mu \psi
_{t^{-}+1}^{-})x_{t^{+}+1} 
\]%
and from $\varphi _{N}^{+}(x)=\varphi _{N}^{-}(x)$ we get $\psi
_{\{t^{+}+1,\ldots ,t^{-}\}}^{-}+\mu \psi
_{t^{-}+1}^{-}=\sum_{1}^{t^{+}}(\psi _{i}^{+}-\psi _{i}^{-})+\lambda \psi
_{t^{+}+1}^{+}$. Rearranging the right-hand term in the above inequality,
and recalling the definition of $\psi _{i}^{\varepsilon }$ gives%
\[
\Delta \leq \sum_{1}^{t^{+}}(\psi _{i}^{+}-\psi
_{i}^{-})(x_{i}-x_{t^{+}+1})=(\theta ^{+}-\theta ^{-})\sum_{1}^{t^{+}}\delta
_{i}(x_{i}-x_{t^{+}+1}). 
\]%
We show finally that the right-hand term above is strictly negative, as
desired.

The sequence $\delta _{i}$ is (strictly) decreasing and initially positive.
As $\delta _{\{1,\ldots ,t^{+}\}}+\lambda \delta _{t^{+}+1}>0$, we have $%
\delta _{\{1,\ldots ,t^{+}\}}>0$. The sequence $\gamma
_{i}=x_{t^{+}+1}-x_{i} $ is positive and (strictly) decreasing. These facts
imply that $\sum_{1}^{t^{+}}\delta _{i}\gamma _{i}$ is strictly positive. Let $%
\delta _{i^{\ast }}$ be the first strictly negative term in the sequence $%
\delta _{i}$: we have $\sum_{1}^{i^{\ast }-1}\delta _{i}\gamma _{i}\geq
\sum_{1}^{i^{\ast }-1}\delta _{i}\gamma _{i^{\ast }}$ as all terms are non-negative and $\gamma _{i}$ is decreasing; also $\sum_{i^{\ast }}^{t^{+}}\delta
_{i}\gamma _{i}>\sum_{i^{\ast }}^{t^{+}}\delta _{i}\gamma _{i^{\ast }}$ as $%
\delta _{i}<0$ and $\gamma _{i}<\gamma _{i^{\ast }}$. Thus $-\Delta
=\sum_{1}^{t^{+}}\delta _{i}\gamma _{i}>\delta _{\{1,\ldots ,t^{+}\}}\gamma
_{i^{\ast }}$.\end{proof}

\section{Proofs for Section~\ref{sec5}}\label{appB}

\subsection{Proof of Lemma~\ref{lm1}}
\begin{proof}[For goods.] The inequality $CR_{n}(\varphi )\leq
PoF_{n}(\varphi )$ is clear. Next, for any $\mathcal{P}\in {\Large \Pi }_{n}$
there exists some $x\in 
\mathbb{R}
_{+}^{N}$ such that%
\[
\frac{\mathbb{E}_{\mu }(\max_{i}X_{i})}{S (\varphi ,\mathcal{P})}\leq 
\frac{\max_{i}x_{i}}{\sum_{i\in N}\varphi _{i}(x)\cdot x_{i}}. 
\]%
This proves $PoF_{n}(\varphi )\leq \sup_{x\in 
	\mathbb{R}
	_{+}^{N}}\frac{\max_{i}x_{i}}{\sum_{i\in N}\varphi _{i}(x)\cdot x_{i}}$. 

Next
we pick an arbitrary $x\in 
\mathbb{R}
_{+}^{N}$ and check the inequality $\frac{\max_{i}x_{i}}{\sum_{i\in
		N}\varphi _{i}(x)\cdot x_{i}}\leq CR_{n}(\varphi )$, thus completing the
proof.
Consider a problem $\mathcal{P}\in {\Large \Pi }_{n}$ that
selects each of the $n!$ permutations of $\frac{1}{\overline{x}}x$ with
equal probability $\frac{1}{n!}$.  {We call a problem \emph{symmetric} if the distribution $\mu$ is symmetric in all variables $x_{i}$.} By the symmetry of the rule $\varphi $ we have $S
(\varphi ,\mathcal{P})=\sum_{i\in N}\varphi _{i}(x)\cdot x_{i}$. It will be
enough to construct a rule $\psi \in \Phi (FS)$ such that $S (\psi ,\mathcal{P%
})=\max_{i}x_{i}$, because $\frac{S (\psi ,\mathcal{P})}{S (\varphi ,%
	\mathcal{P})}\leq CR_{n}(\varphi )$. To this end, we note that the
Utilitarian rule $\varphi^{ut}$ violates FS in general (see the example in Section~\ref{subsect_benchmark_rules}) but not if the problem $\mathcal{P}$ is symmetric.\footnote{%
	Indeed, $\mathbb{E}_{\mu }(X_{1})\leq \mathbb{E}_{\mu }(\max_{i}X_{i})=S (\varphi^{ut},%
	\mathcal{P})=\sum_{i}\mathbb{E}_{\mu }(\varphi^{ut}_{i}(X)\cdot
	X_{i})=n\mathbb{E}_{\mu }(\varphi^{ut}_{1}(X)\cdot X_{1}).$ 
} Thus, we can pick a $\psi $ that is equal to $\varphi^{ut}$ for symmetric
problems, and satisfies FS elsewhere.

\smallskip
\noindent\textit{For bads.} The argument is similar and therefore omitted.
\end{proof}

\subsection{Proof of Proposition~\ref{prop4}}
\begin{proof}[Statement $i)$.] Pick $\varphi \in \Phi_{ind}(FS)$ and $\mathcal{%
	P}\in {\Large \Pi }_{n}$. The FS property implies%
\[
S (\varphi ,\mathcal{P})=\sum_{i\in N}\mathbb{E}_{\mu }(\varphi
_{i}(X)\cdot X_{i})\geq \frac{1}{n}\sum_{i\in N}\mathbb{E}_{\mu }(X_{i})\geq 
\frac{1}{n}\mathbb{E}_{\mu }(\max_{i}X_{i}) 
\]%
and the first claim follows. If $\varphi $ is the Equal Split rule, the
first inequality shown above is an equality, and the second one is an equality if
the random variable $X$ is uniform over the coordinate profiles $e^{\{i\}}$.

\smallskip
\noindent\textit{Statement $ii)$.}  By Lemma~\ref{lm1} we must evaluate $\sup_{x\in 
	\mathbb{R}
	_{+}^{N}\diagdown \{0\}}\frac{\sum_{i\in N}x_{i}}{\sum_{i\in N}x_{i}^{2}}%
\max_{i}x_{i}$. By rescaling $x$ we can assume that $x_{1}=1=\max_{i\geq 2}x_{i}$; then we must show that%
\[
\sup \frac{1+\sum_{2}^{n}x_{i}}{1+\sum_{2}^{n}x_{i}^{2}}=\frac{\sqrt{n}+1}{2},
\]%
where the supremum is on all $x_{2},\ldots ,x_{n}\in \lbrack 0,1]$. The argument is  straightforward and therefore  omitted.

\smallskip
\noindent\textit{Statement $iii)$.} We fix $\theta,\ 0<\theta \leq 1$, set $%
N=\{1,\ldots ,n\}$, and rewrite inequalities (\ref{10}) as

\[
\varphi _{i}^{\theta }(x)\geq \max \left\{\left(\frac{1}{n}+\frac{\theta }{n-1}\right)-\frac{%
	\theta }{n(n-1)}\frac{x_{N}}{x_{i}},\ 0\right\}\mbox{ for all }i\mbox{ and }x\in 
\mathbb{R}
_{+}^{N}.
\]%
By Lemma~\ref{lm1} we must evaluate the smallest feasible value of $\frac{1}{x^{(n)}}\{\sum_{i=1}^{n}\varphi _{i}^{\theta }(x)\cdot x_{i}\}$ in $%
\mathbb{R}
_{+}^{N}$. This function is continuous in $x$ (even though 
$\varphi ^{\theta }$ itself is not in those profiles where several agents
have the highest utility), and so it will be enough to compute the infimum of
this ratio for profiles $x$ such that $x_{i}<x_{n}$ for all $i\leq n-1$.

We first compute the desired lower bound when $(\frac{1}{n}+\frac{\theta }{%
	n-1})-\frac{\theta }{n(n-1)}\frac{x_{N}}{x_{i}}\geq 0$ for all $i$, so that
all agents $i\leq n-1$ get exactly this share and agent $n$ gets%
\[
\varphi _{n}^{\theta }(x)=1-\sum_{i=1}^{n-1}\varphi _{i}^{\theta }(x)=\frac{1%
}{n}-\theta +\frac{\theta }{n(n-1)} \left(\left(\sum_{i=1}^{n-1}\frac{1}{x_{i}}%
\right)x_{n}+n-1+\sum_{\{i,j\}\subset \{1,\ldots ,n-1\}}\left(\frac{x_{i}}{x_{j}}+\frac{%
	x_{j}}{x_{i}}\right) \right). 
\]

On the right-hand side, if we fix the sum $\sum_{i=1}^{n-1}x_{i}$, the first
sum is minimal when all utilities are equal; the second sum is also minimal
and equal to $(n-1)(n-2)$ when utilities are equal. It is also clear that
for $i,j\leq n-1$ the sum $\varphi _{i}^{\theta }(x)\cdot x_{i}+\varphi
_{j}^{\theta }(x)\cdot x_{j}$ is constant when we equalize $x_{i}$ and $%
x_{j} $ while keeping their sum constant. Thus, we can assume that $x_{i}=y$
for $1\leq i\leq n-1$, so that the share of agent $n$ is%
\[
\varphi _{n}^{\theta }(x)=\frac{1}{n}-\theta +\frac{\theta }{n}\left(\frac{x_{n}}{%
	y}+n-1\right)=\frac{1}{n}(1-\theta )+\frac{\theta }{n}\frac{x_{n}}{y}. 
\]%
Then we compute%
\[
\frac{1}{x_{n}}\left(\sum_{i=1}^{n}\varphi _{i}^{\theta }(x)\cdot x_{i}\right)=\varphi
_{n}^{\theta }(x)+(n-1)\frac{y\cdot \varphi _{1}^{\theta }(x)}{x_{n}}=\frac{1}{n}%
\left((1-2\theta )+\theta \frac{x_{n}}{y}+(n-1+\theta )\frac{y}{x_{n}}%
\right) 
\]%
and the minimum in $x_{n},y$ of this expression is achieved for $\frac{x_{n}%
}{y}=\big(\frac{(n-1+\theta )}{\theta }\big)^{\frac{1}{2}}$ (which is greater than $1$,
as needed) and its value is%
\[
\frac{1}{n}\left((1-2\theta )+2\sqrt{(n-1+\theta )\theta }\right), 
\]%
as stated. Clearly it is decreasing in $\theta $.

It remains to consider the case where for some $i^{\ast }\leq n-1$ we have,
for all $i\leq i^{\ast }-1$ and all $j\geq i^{\ast }$,%
\[
\left(\frac{1}{n}+\frac{\theta }{n-1}\right)-\frac{\theta }{n(n-1)}\frac{x_{N}}{x_{i}}%
<0\leq \left(\frac{1}{n}+\frac{\theta }{n-1}\right)-\frac{\theta }{n(n-1)}\frac{x_{N}}{%
	x_{j}}. 
\]
Observe that if we decrease $x_{i}$ to zero for all $i\leq i^{\ast }-1$
without changing other coordinates, the share of each agent $j,i^{\ast }\leq
j\leq n-1,$ increases (strictly if some $x_{i}$ is positive), while that of
agent $n$ decreases; therefore the ratio $\frac{1}{x^{(n)}}%
\{\sum_{i=1}^{n}\varphi _{i}^{\theta }(x)\cdot x_{i}\}$ decreases. Thus, it
is enough to assume $x_{i}=0$ for all $i\leq i^{\ast }-1$. Computing the
share of agent $n$ and the total utility $\sum_{i=1}^{n}\varphi _{i}^{\theta
}(x)\cdot x_{i}$ is then more involved but very similar, and the argument
that we can assume $x_{i}=y$ for $i^{\ast }\leq i\leq n-1$ is unchanged. Therefore,%
\[
\varphi _{n}^{\theta }(x)=\frac{i^{\ast }}{n}\left(1-\frac{n-i^{\ast }}{n-1}%
\theta \right)+\frac{n-i^{\ast }}{n(n-1)}\theta \frac{x_{n}}{y}, 
\]%
\[
\frac{1}{x_{n}}\left(\sum_{i=1}^{n}\varphi _{i}^{\theta }(x)\cdot x_{i}\right)=\frac{%
	i^{\ast }}{n}-\frac{(n-i^{\ast })(i^{\ast }+1)}{n(n-1)}\theta +\frac{%
	n-i^{\ast }}{n(n-1)}\left(\theta \frac{x_{n}}{y}+(n-1+i^{\ast }\theta )%
\frac{y}{x_{n}}\right) 
\]%
of which the minimum in $x_{n},y$ is%
\[
\frac{i^{\ast }}{n}-\frac{(n-i^{\ast })}{n(n-1)}\left((i^{\ast
}+1)\theta -2\sqrt{(n-1+i^{\ast }\theta )\theta }\right) 
\]%
and this quantity is increasing in $i^{\ast }$ because $(i^{\ast }+1)\theta -2%
\sqrt{(n-1+i^{\ast }\theta )\theta }$ does. Therefore, the worst case is  $%
i^{\ast }=1$, and we are done.

\smallskip
\noindent\textit{Statement $iv)$.}  Clearly $\inf_{\varphi \in \Phi(FS)}PoF_{n}(\varphi )\leq \inf_{\varphi \in \Phi_{ind} (FS)}PoF_{n}(\varphi )\leq
PoF_{n}(\varphi ^{1})$, and so the inequality $\inf_{\varphi \in \Phi(FS)}PoF_{n}(\varphi )\leq \frac{n}{2\sqrt{n}-1}$ follows from Lemma~\ref{lm1} and
statement $iii)$.

Next we fix $n,m,$ such that $1\leq m\leq n-1$ and consider the problem $\mathcal{P}%
(n,m)\in {\Large \Pi }_{n}$ with $n$ agents, $m$ equiprobable states, and the utilities defined in Table~\ref{tab_worst_case}.
\begin{table}[h!]
	\begin{center}
		\caption{States and utilities}		
		{\begin{tabular}{c|cccc}\label{tab_worst_case}
				state & $\omega _{1}$ & $\omega_2$  &$\ldots$ & $\omega _{m}$\\
				\hline  
				probability & $1/m$ & $1/m$ & $\ldots$ & $1/m$ \\ 
				\hline
				$X_{1}$ & $m$ & $0$ & $\ldots$ & $0$ \\
				$X_{2}$ & $0$ & $m$ & $\ldots$ & $0$ \\
				$\vdots$ & $0$ & $0$ & $\ddots$ &$0$\\
				$X_{m}$ & $0$ & $0$ & $\ldots$ &$m$ \\ 
				\hline
				$X_{m+1}$ & 1 & 1 & $\ldots$ & 1 \\ 
				$\vdots$ &  1 & 1 & $\ldots$ & 1 \\ 
				$X_{n}$ & 1 & 1 & $\ldots$ & 1 
		\end{tabular}}
	\end{center}
\end{table}

Let $N_{1}$ be the set of the $m$ 
``single-minded'' agents and $N_{2}$ be the set of the other $n-m$
``indifferent'' agents. Fix an arbitrary {
	prior-dependent} rule $\varphi \in \Phi(FS)$ and let $\mathbb{E}_{\mu }(Y_{i})=%
\mathbb{E}_{\mu }(\varphi_{i}^{\mu}(X)\cdot X_{i})$ be the expected
utility of agent $i$.

We denote by  $\lambda _{k}$ the total share $\varphi$ gives to $N_{2}$ at state $%
\omega _{k}$. Then the identity $\mathbb{E}_{\mu }(Y_{N_{2}})=\frac{1}{m}%
\sum_{k=1}^{m}\lambda _{k}$ and Fair Share imply $\sum_{k=1}^{m}\lambda
_{k}\geq \frac{m(n-m)}{n}$. If $\varphi$ gives the remaining shares to single-minded agent $k$ in state $\omega _{k}$, then $\mathbb{E}_{\mu }(Y_{N_{1}})=%
\frac{1}{m}\sum_{k=1}^{m}(1-\lambda _{k})m=m-\sum_{k=1}^{m}\lambda _{k}$.
This is the best $\varphi$ can do for the utilitarian objective. Compute%
\[
\mathbb{E}_{\mu }(Y_{N})=\left(m-\sum_{k=1}^{m}\lambda _{k}\right)+\left(\frac{1}{m}%
\sum_{k=1}^{m}\lambda _{k}\right)=m-\frac{m-1}{m}\sum_{k=1}^{m}\lambda _{k}\leq 
\]%
\[
\leq m-\frac{(m-1)(n-m)}{n}=\frac{m^{2}}{n}-\frac{m}{n}+1 
\]%
\[
\Longrightarrow \left( \frac{\mathbb{E}_{\mu }(\max_{i}X_{i})}{\mathbb{E}%
	_{\mu }(Y_{N})}\right)^{-1}=\frac{\mathbb{E}_{\mu }(Y_{N})}{m}\leq \frac{m}{n}+%
\frac{1}{m}-\frac{1}{n}. 
\]%
The minimum of $\frac{m}{n}+\frac{1}{m}-\frac{1}{n}$ over real numbers is
achieved for $m=\sqrt{n}$, and is worth $\frac{2}{\sqrt{n}}-\frac{1}{n}%
=(CR_{n}(\varphi ^{1}))^{-1}$. As $m$ is an integer and $m\rightarrow f(m)=%
\frac{m}{n}+\frac{1}{m}$ is convex, the minimum over integers is at most $%
\alpha =\max \{f(\sqrt{n}+\frac{1}{2}),f(\sqrt{n}-\frac{1}{2})\}$. Routine
computations show that $\alpha \leq \frac{2}{\sqrt{n}}+\frac{1}{2n}$; therefore $\left(%
\frac{\mathbb{E}_{\mu }(\max_{i}X_{i})}{\mathbb{E}_{\mu }(Y_{N})}\right)^{-1}\leq 
\frac{2\sqrt{n}-\frac{1}{2}}{n}$ and the proof is complete.\end{proof}

\subsection{Proof of Proposition~\ref{prop5}}
\begin{proof}[Statement $i)$.] If $\varphi $ is the Equal Split rule,
then $\frac{1}{\min_{i}x_{i}}(\sum_{i\in N}\varphi _{i}(x)\cdot x_{i})=\frac{%
	x_{N}}{n\cdot \min_{i}x_{i}}$ for all $x\in 
\mathbb{R}
_{+}^{N}$. This ratio is clearly unbounded, and the claim follows by
Lemma~\ref{lm1}.

Recall  that, by definition,  {the Proportional rule $\varphi^{pro}$ coincides with the Utilitarian rule at any profile} $x\in 
\mathbb{R}
_{+}^{N}$ with at least one zero coordinate. For $x\gg 0$ we have $\frac{1}{%
	\min_{i}x_{i}}(\sum_{i\in N}\varphi _{i}^{pro}(x)\cdot x_{i})=\frac{1}{%
	\min_{i}x_{i}}\frac{n}{\sum_{i\in N}\frac{1}{x_{i}}}=\frac{\widetilde{x}}{%
	\min_{i}x_{i}}$, where $\widetilde{x}$ is the harmonic mean of $x_{i}$.
The inequality $\widetilde{x}\leq n\min_{i}x_{i}$ is always true, and  {asymptotically}
becomes an equality when $x_{1}=\min_{i}x_{i}$, and all other coordinates
are equal and go to infinity. Therefore, the $CR_n(\varphi ^{pro})$ is indeed $n$.

\smallskip
\noindent\textit{Statement $ii)$.} The lower bound follows from the lower
bound on $\inf_{\varphi \in \Phi(FS)}PoF_{n}(\varphi )$ (statement $iii$ proven
below) and from $CR_{n}(\varphi ^{1})=PoF_{n}(\varphi ^{1})\geq \inf_{\varphi
	\in \Phi(FS)}PoF_{n}(\varphi )$.

To prove the upper bound\textit{ }$PoF_{n}(\varphi ^{1})\leq \frac{n}{4}+%
\frac{5}{4}$, we fix an arbitrary profile $x$ and majorize $\frac{1}{%
	\min_{i}x_{i}}(\sum_{i\in N}\varphi _{i}^{1}(x)\cdot x_{i})$. Because $%
\varphi ^{1}$ is homogeneous of degree zero and symmetric, and $\varphi ^{1}$
{coincides with the Utilitarian rule} if $x_{1}=0$, we can without loss of generality assume that $x_{1}=1$ and $x_{i}$ is 
 weakly increasing in~$i$. We must bound $U_{N}(x)=\sum_{i\in N}\varphi
_{i}^{1}(x)\cdot x_{i}$.  {By the continuity of $U_N(x)$, we can assume that none of the coordinates of $x$ are equal, i.e., that $x_i$ is strictly increasing.}

By the definition of $\varphi^{1}$ there exists an  {index $\widetilde{t}$} such that%
\[
\frac{1}{n(n-1)}\sum_{i=1}^{\widetilde{t}}\frac{x_{N\diagdown \{i\}}}{x_{i}}\leq 1<\frac{1%
}{n(n-1)}\sum_{i=1}^{\widetilde{t}+1}\frac{x_{N\diagdown \{i\}}}{x_{i}} 
\]%
and $\varphi _{i}^{1}(x)=\frac{1}{n(n-1)}\frac{x_{N\diagdown \{i\}}}{x_{i}}$ for 
$i\leq \widetilde{t}$.

We set $\Delta =n(n-1)-\sum_{i=1}^{\widetilde{t}}\frac{x_{N\diagdown \{i\}}}{x_{i}}$, $%
\Delta \geq 0$, and develop $U_{N}(x)$ as follows:%
\[
n(n-1)U_{N}(x)=\sum_{i=1}^{\widetilde{t}}x_{N\diagdown \{i\}}+\Delta
x_{\widetilde{t}+1}=(\widetilde{t}-1)\sum_{i=1}^{\widetilde{t}}x_{i}+\widetilde{t}\sum_{j=\widetilde{t}+1}^{n}x_{j}+\Delta x_{\widetilde{t}+1}. 
\]

Suppose that we replace each $x_{i},\ 2\leq i\leq \widetilde{t}$, by their average $y=\frac{1}{\widetilde{t}-1}%
\sum_{i=2}^{\widetilde{t}}$, ceteris paribus: this will decrease the total weight
given by $\varphi ^{1}$ to these coordinates, which is $\frac{x_{N}}{n(n-1)}%
(\sum_{2}^{\widetilde{t}}\frac{1}{x_{i}})$, and it will increase the weight to coordinates $%
x_{\widetilde{t}+1}$ and beyond. Therefore, this move increases $U_{N}(x)$, and so we can
assume that these $\widetilde{t}-1$ coordinates are all equal to $y$. We also set $%
\sum_{j=\widetilde{t}+1}^{n}x_{j}=w$. Now we try to bound%
\[
n(n-1)U_{N}(x)=(\widetilde{t}-1)(1+(\widetilde{t}-1)y)+\widetilde{t}w+\Delta x_{\widetilde{t}+1} 
\]%
under the constraints%
\[
\Delta =n(n-1)+\widetilde{t}-(1+(\widetilde{t}-1)y+w)\left(1+\frac{\widetilde{t}-1}{y}\right)\geq 0\mbox{ ; }0\leq \Delta
x_{\widetilde{t}+1}\leq 1+(1-\widetilde{t})y+w\mbox{ ; }w\geq (n-\widetilde{t})y, 
\]%
where we infer the second inequality from the fact that $\Delta \leq \frac{%
	x_{N\diagdown \{(\widetilde{t}+1)\}}}{x_{\widetilde{t}+1}}$ and the third one from the fact that the
coordinates of $x$ are weakly increasing. These inequalities imply%
\[
n(n-1)U_{N}(x)\leq \widetilde{t}(1+(\widetilde{t}-1)y)+(\widetilde{t}+1)w, 
\]%
\[
(1+(\widetilde{t}-1)y+w)\left(1+\frac{\widetilde{t}-1}{y}\right)\leq n(n-1)+\widetilde{t}\Longrightarrow \left(1+\frac{\widetilde{t}-1}{y}\right)w\leq
n(n-1)-(\widetilde{t}-1)\left(y+\frac{1}{y}\right)+(\widetilde{t}-1)-(\widetilde{t}-1)^{2} 
\]%
\[
\Longrightarrow w\leq (n(n-1)+\widetilde{t}-1)\frac{y}{y+\widetilde{t}-1}-(\widetilde{t}-1)y. 
\]%
Combining $w\geq (n-\widetilde{t})y$ and the upper bound above gives%
\[
(n-\widetilde{t})y\leq (n(n-1)+\widetilde{t}-1)\frac{y}{y+\widetilde{t}-1}-(\widetilde{t}-1)y\Longrightarrow y+\widetilde{t}-1\leq n+\frac{\widetilde{t}-1}{n-1%
}\leq n+1. 
\]%
Next we combine the upper bound on $n(n-1)U_{N}(x)$ with that on $w$:%
\[
n(n-1)U_{N}(x)\leq \widetilde{t}(1+(\widetilde{t}-1)y)+(\widetilde{t}+1)(n(n-1)+\widetilde{t}-1)\frac{y}{y+\widetilde{t}-1}-(\widetilde{t}+1)(\widetilde{t}-1)y= 
\]%
\[
=\widetilde{t}-(\widetilde{t}-1)y+(\widetilde{t}+1)(n(n-1)+\widetilde{t}-1)\frac{y}{y+\widetilde{t}-1}. 
\]%
We now majorize the above upper bound in the two real variables $\widetilde{t},y$ such
that $y+\widetilde{t}\leq n+2$. Observe first that this bound is increasing in $y$ because its
derivative has the sign of $\frac{(\widetilde{t}+1)(n(n-1)+\widetilde{t}-1)}{(y+\widetilde{t}-1)^{2}}-1$ and $\frac{%
	(\widetilde{t}+1)(n(n-1)+\widetilde{t}-1)}{(y+\widetilde{t}-1)^{2}}\geq \frac{3(n^{2}-n+1)}{(n+1)^{2}}$. Thus, we can
take $y+\widetilde{t}=n$ and use the inequality $\frac{\widetilde{t}+1}{n+1}\leq 1$ to deduce the
bound%
\[
n(n-1)U_{N}(x)\leq \widetilde{t}+\frac{n(n-1)(\widetilde{t}+1)y}{n+1}+(\widetilde{t}-1)y\left(\frac{\widetilde{t}+1}{n+1}-1\right)\leq n+%
\frac{n(n-1)}{n+1}(\widetilde{t}+1)(n+2-\widetilde{t}). 
\]%
The maximum in $\widetilde{t}$ of $(\widetilde{t}+1)(n+2-\widetilde{t})$ is $\frac{(n+3)^{2}}{4}$ for $\widetilde{t}=\frac{%
	n+1}{2}$; therefore%
\[
\Longrightarrow U_{N}(x)\leq \frac{1}{n-1}+\frac{(n+3)^{2}}{4(n+1)}=\frac{n}{%
	4}+\frac{5}{4}-\frac{2}{n^{2}-1},
\]%
completing the proof of statement $ii)$.

\smallskip
\noindent\textit{Statement $iii)$.} \ \\ 
\noindent\textit{Step 1. Lower bound on $\inf_{\varphi \in \Phi
		(FS)}PoF_{n}(\varphi )$.} Consider the normalized problem $\mathcal{P}$ with
two equally probable states $\omega ,\omega ^{\prime }$, and the
corresponding profiles of disutilities%
\[
x_{1}=\frac{4}{n+1},\ \ x_{i}=2\ \mbox{ for }\ 2\leq i\leq n\mbox{ ; }\ \ \
x_{1}^{\prime }=2\cdot\frac{n-1}{n+1},\ \ x_{i}^{\prime }=0\ \mbox{ for }\ 2\leq
i\leq n. 
\]%
Without the FS constraint the total disutility is minimized by giving to agent $%
1 $ the whole bad in state $\omega $, and no share at all in state $\omega
^{\prime }$, so that $\mathbb{E}_{\mu }(\min_{i}X_{i})=\frac{2}{n+1}$. The
FS constraint caps the share of agent $1$ at $\frac{n+1}{2n}$ in state $%
\omega $ and so at least $\frac{n-1}{2n}$ goes to the other agents and the expected
total disutility is at least $\frac{1}{n}+\frac{1}{2}\frac{n-1}{2n}2=\frac{%
	n+1}{2n}$. Therefore, for any $\varphi \in \Phi (FS)$ we have%
\[
\frac{S (\varphi ,\mathcal{P})}{\mathbb{E}_{\mu }(\min_{i}X_{i})}\geq 
\frac{(n+1)^{2}}{4n}=\frac{n}{4}+\frac{1}{2}+\frac{1}{4n}. 
\]%

\smallskip
\noindent\textit{\textit{Step 2}. \textit{Upper bound on }$\inf_{\varphi \in \Phi
		(FS)}PoF_{n}(\varphi )$.} 
 Fix a problem $\mathcal{P}\in {\LARGE \Pi }_{n}$
and let $\mathcal{C}$ denote the compact convex set of the disutility profiles
feasible by some prior-dependent rule. Then $\mathcal{C}$ contains the
simplex $\Delta (N)$ because the rule giving the object always to agent $i$
achieves the unit vector $e^{i}$. Let $x\in \mathcal{C}$ achieve the
smallest total disutility in $\mathcal{C}$: $x_{N}=\mathbb{E}_{\mu
}(\min_{i}X_{i})$. We must construct a profile $y$ in $\mathcal{C}$ satisfying
FS and such that%
\[
y_{N}\leq \left(\frac{n}{4}+\frac{1}{2}+\frac{1}{4n}\right)x_{N}. 
\]%
If $x$ satisfies FS we can take $y=x$ and if $x_{N}=1$ we take $y=\frac{1}{n}%
e^{N}$, the center of the simplex. Otherwise some coordinates of $x$ are
above $\frac{1}{n}$; upon relabeling coordinates we have%
\[
x_{1}\geq \ldots \geq x_{t}>\frac{1}{n}\geq x_{t+1}\geq \ldots \geq x_{n} 
\]%
and we keep in mind $x_{N}<1$. We use below the notation $K=\{1,\ldots ,t\}$
and $L=\{t+1,\ldots ,n\}$.

We wish to choose $y=\lambda x+\lambda ^{\prime }\tau $, a convex
combination of $x$ and some $\tau \in \Delta (N)$, such that%
\[
\lambda x_{k}+\lambda ^{\prime }\tau _{k}=\frac{1}{n}\text{ for }1\leq k\leq
t\text{ ; }\lambda x_{\ell }+\lambda ^{\prime }\tau _{\ell }\leq \frac{1}{n}%
\text{ for }t+1\leq \ell \leq n. 
\]%
For any $\lambda \in \lbrack 0,1]$ such that $\lambda x_{1}\leq \frac{1}{n}$%
, each one of the $t$ equalities shown above defines $\tau _{k}$ in $[0,1]$
(because $\lambda x_{k}+\lambda ^{\prime }\geq \frac{\lambda }{n}+\lambda
^{\prime }\geq \frac{1}{n}$) and their sum $\tau _{K}$. We can then find non-negative numbers $\tau _{\ell }$ satisfying the last $n-t$ inequalities as well
as $\tau _{L}=1-\tau _{K}$ {iff }$\tau _{K}\leq 1$ and $\lambda
x_{L}+\lambda ^{\prime }\tau _{L}\leq \frac{n-t}{n}$.

By construction, $\lambda ^{\prime }\tau _{K}=\frac{t}{n}-\lambda x_{K}$ and  so
the last two inequalities are%
\[
\tau _{K}\leq 1\Longleftrightarrow \frac{t}{n}-\lambda x_{K}\leq \lambda
^{\prime }\Longleftrightarrow \lambda (1-x_{K})\leq \frac{n-t}{n}, 
\]%
\[
\lambda x_{L}+\lambda ^{\prime }(1-\tau _{K})\leq \frac{n-t}{n}%
\Longleftrightarrow \lambda x_{L}+\lambda x_{K}+\lambda ^{\prime }\leq
1\Longleftrightarrow x_{N}\leq 1. 
\]

The latter inequality is true. The former is a consequence of $\lambda
x_{1}\leq \frac{1}{n}$ because by the definition of $K$ we have $%
(n-t)x_{1}+x_{K}>1$, implying $\frac{1}{nx_{1}}<\frac{n-t}{n(1-x_{K})}$. Therefore $%
\lambda x_{1}\leq \frac{1}{n}$ is the only constraint on the choice of $%
\lambda $.

We choose $\lambda $ to minimize%
\[
\frac{y_{N}}{x_{N}}=\frac{\lambda x_{N}+\lambda ^{\prime }}{x_{N}}=\lambda +%
\frac{1-\lambda }{x_{N}}=\frac{1}{x_{N}}-\lambda (\frac{1}{x_{N}}-1), 
\]%
which is decreasing in $\lambda $. Therefore we pick $\lambda =\frac{1}{nx_{1}}$ to get%
\[
\frac{y_{N}}{x_{N}}=\frac{nx_{1}-1}{nx_{1}x_{N}}+\frac{1}{nx_{1}}\leq \frac{%
	nx_{1}-1}{nx_{1}^{2}}+\frac{1}{nx_{1}}=\frac{n+1}{n}\frac{1}{x_{1}}-\frac{1}{%
	nx_{1}^{2}}. 
\]%
We leave it to the reader to check that the maximum of the right-hand term for $%
x_{1}\in \lbrack \frac{1}{n},1]$ is reached at $x_{1}=\frac{2}{n+1}$ and is
precisely $\frac{n}{4}+\frac{1}{2}+\frac{1}{4n}$. 
\end{proof}

Interestingly, the PoF we just computed is the inverse of the
``price of maximin fairness'' for classic bargaining problems (corresponding in our
model to the division of a good); see~\citet[Theorem~1]{Bertsimas2011}.

\section{Asymptotic results and missing proofs for Section~\ref{sec6}}\label{appC}

\subsection{A good}
\begin{proposition}\label{prop6}
 \textit{Fix a distribution }$\nu $%
\textit{ of }$X_{i}$\textit{ with }$E_{\nu }X_{1}=1$\textit{ and }$E_{\nu
}(X_{1})^{\beta }<\infty $\textit{ for some }$\beta >2$\textit{. Consider a
	problem }$\mathcal{P}_{n}(\nu )$\textit{ with }$n$\textit{ agents and }$%
\mu =\otimes _{i=1}^{n}\nu $\textit{. Then the  ratio for
	the TH rule }$\varphi ^{\theta }$\textit{, }$\theta \in (0,1]$\textit{,
	satisfies} 
\begin{equation}
\pi(\varphi ^{\theta },\mathcal{P}_{n}(\nu ))=\frac{1}{1-\mathbb{E}_{\nu }\left(
	1+\theta -\frac{\theta }{X_{1}}\right) _{+}+\frac{\mathbb{E}_{\nu }\left(
		X_{1}(1+\theta )-\theta \right) _{+}}{\mathbb{E}_{\mu }(X^{(n)})}}\left(1
+O\left( \frac{1}{n^{\frac{1}{2}-\frac{1}{\beta }}}\right)\right),
\label{eq_asymp_good_iid}
\end{equation}%
\textit{for a large number of agents\footnote{%
		$a_{n}=O(b_{n})$ if there exist $n_{0}$ and $C>0$ such that $%
		|a_{n}|\leq C|b_{n}|$ for all $n\geq n_{0}$.} }$n$\textit{. Here }$(y)_{+}$\textit{
	denotes }$\max \{y,0\}$\textit{.}
\end{proposition}

\medskip Note that the only dependence on $n$ in formula~(\ref%
{eq_asymp_good_iid}) is through the expected value of $X^{(n)}=%
\max_{i=1,..,n}X_i$ and the error term.

\begin{proof}[Proof of Proposition~\ref{prop6}.]

For simplicity, we assume that $\theta=1$ (proofs for other values
of $\theta$ follow the same logic).

By the definition of the TH rule $\varphi ^{1}$ we can represent {the
social welfare} as 
\[
\sum_{i}X_{i}\varphi _{i}^{1}(X)=\sum_{i=1}^{n}X_{i}\left( \frac{2}{n}-\frac{%
	X_{N}-X_{i}}{n(n-1)X_{i}}\right) _{+}+X^{(n)}\left(
1-\sum_{i=1}^{n}\left( \frac{2}{n}-\frac{X_{N}-X_{i}}{n(n-1)X_{i}}\right)
_{+}\right)= 
\]%
\[
=A+X^{(n)}-B. 
\]%
Consider the contribution of $A$ first. Since all $X_{i}$ have the same
distribution, it follows that  $\mathbb{E}_{\mu }A=\mathbb{E}_{\mu }\left( 2X_{1}-\frac{%
	\sum_{j\neq 1}X_{j}}{n-1}\right) _{+}$. Let us show that $\Delta _{0}=%
\mathbb{E}_{\mu }(A)-\mathbb{E}_{\nu }\left( 2X_{1}-1\right) _{+}$ is small.
The function $(\,\cdot \,)_{+}$ is Lipschitz with constant one; thus by the
Cauchy inequality and the independence of $X_{j}$, we have
\[
\left\vert \Delta _{0}\right\vert \leq \mathbb{E}_{\mu }\left(\left\vert 1-\frac{%
	\sum_{j\neq 1}X_{j}}{n-1}\right\vert\right) =\frac{1}{n-1}\mathbb{E}_{\mu
}\left(\left\vert {\sum_{j\neq 1}(X_{j}-1)}\right\vert \right)\leq
\]%
\[
\leq \frac{1}{n-1}\sqrt{\mathbb{E}_{\mu }\left( \sum_{j\neq
		1}(X_{j}-1)\right) ^{2}}=\frac{\sqrt{\mathbb{V}_{\nu }(X_{1})}}{\sqrt{n-1}}%
=O\left( \frac{1}{\sqrt{n}}\right) 
\]%
if the variance $\mathbb{V}_{\nu }$ of $X_{1}$ is finite.

Now we will check that $\mathbb{E}_\mu (B)$ is close to $\mathbb{E}_\mu
(X^{(n)})\cdot \mathbb{E}_\nu \left((2-1/X_1)_+\right)$ (as if $X^{(n)}$ is independent of $X_i$
and $\sum {X_j}$ approximately equals its expectation). This is done in two
steps:

\begin{itemize}
	\item \emph{Step 1:} Prove that $\mathbb{E}_{\mu }(B)$ does not change much if we substitute $%
	(2-1/X_{1})_{+}$ for $(2-\sum_{j\neq 1}X_{j}/(n-1)X_{1})_{+}$.
	
\item \emph{Step 2:} Prove that the random variables $X^{(n)}$ and $(2-1/X_{1})_{+}$
	can be decoupled; the expected value of the product is close to the product
	of expectations.
\end{itemize}

\smallskip
\noindent\textit{Step 1: \  $\sum_{j\neq 1}X_{j}/(n-1)$ can be
	replaced by its expectation.}

Since $X_j$ are independent and identically distributed we have 
\[
\mathbb{E}_\mu (B) = \mathbb{E}_\mu \left(X^{(n)} \left(2-\frac{\sum_{j\ne 1}X_j}{
	(n-1)X_1}\right)_+\right)= \mathbb{E} \left(X^{(n)} \left(2-\frac{1}{X_1}\right)_+\right) +
\Delta_1, 
\]
where 
\[
\Delta_1=\mathbb{E}_\mu \left(X^{(n )}\left(\left(2-\frac{\sum_{j\ne 1}X_j}{(n-1)X_1%
}\right)_+ - \left(2-\frac{1}{X_1}\right)_+\right)\right)=\mathbb{E}_\mu \left(X^{(n)}
h(X) \right).
\]
Consider two cases depending on how far  the sum $\sum_{j\ne 1} X_j$ is from
its expected value. Let $Q$ be the event that $\left|\frac{\sum_j X_j}{n-1}%
-1 \right|>\frac{1}{2}$, $\overline{Q}$ its complement, and $1_Q$, $1_{\overline{Q}}$  their indicator functions. Then the probability  $\mathbb{P}_\mu(Q)=\mathbb{E}(1_Q)$ is at
most $\frac{8 \mathbb{V}_\nu (X_1)}{n-1}$ by the Markov inequality. Let us
represent $\Delta_1$ as $\mathbb{E}_\mu \left(X^{(n)}h(X)1_Q\right)+\mathbb{E}_\mu
\left(X^{(n)}h(X)1_{\overline{Q}}\right)$. For the first term, we use the estimate $h\leq 2$ and
then apply the Cauchy inequality: 
\[
\mathbb{E}_\mu \left(X^{(n)}|h(X)|1_Q\right) \leq 2\mathbb{E}_\mu \left(X^{(n)}1_Q\right)\leq \sqrt{\mathbb{E}_\mu \left(|X^{(n)}|^{2}\right)}%
\sqrt{\mathbb{P}_\mu(Q)}.
\]
To bound the second term, consider the following inequality for $y,z\leq 2$: $%
\big||y|_+ - |z|_+\big|\leq (1_{y\geq 0}+1_{z\geq 0})\cdot \big|y-z\big|. $ Applying it
to $h$ we get 
\[
|h(x)|\leq \left(1_{\left\{\frac{1}{x_1}\leq \frac{2(n-1)}{\sum_{j\ne 1 }x_j}
	\right\}} + 1_{\left\{\frac{1}{x_1}\leq 2\right\}} \right) \left|\frac{
	\sum_{j\ne 1} x_j}{(n-1)x_1} - \frac{1}{x_1} \right|. 
\]
For $x\in \overline{Q}$, the function $h$ is non-zero only if $\frac{1}{x_1}\leq 
\frac{4}{3}$. Thus, for such $x$, we have $|h(x)|\leq \frac{8}{3}\left|\frac{%
	\sum_{j\ne 1} (x_j-1)}{n-1} \right|. $ Finally, we get 
\[
\mathbb{E}_\mu \left( X^{(n)}|h(X)|1_{\overline{Q}}\right)\leq \frac{8}{3(n-1)}\mathbb{E}_\mu
\left(X^{(n)}\left|\sum_{j\ne 1} (X_j-1)\right|\right)\leq \frac{8}{3(n-1)} \sqrt{\mathbb{%
		E }_\mu \left(|X^{(n)}|^{2}\right)} \sqrt{ \mathbb{E} \left(\left(\sum_{j\ne 1} (X_j-1)\right)^2\right). 
} 
\]
Combining all the estimates together, we see that $|\Delta_1|= O\left(\frac{%
	\sqrt{\mathbb{E}_\mu \left(|X^{(n)}|^2\right)} }{\sqrt{n}}\right). $ We will estimate $%
\mathbb{E}_\mu \left(|X^{(n)}|^2\right)$ at the end of the proof.

\smallskip
\noindent\textit{Step 2: Decouple $X^{(n)}$ and $(2-1/X_{1})_{+}$.}

We proved that $B$ is close to $\mathbb{E}_{\mu
}\left(X^{(n)}(2-1/X_{1})_{+}\right)$. Now we want to decouple the two factors and
show that $B$ is close to $\mathbb{E}_{\mu }\left(X^{(n)}\right)\cdot \mathbb{E}_{\nu
}\left((2-1/X_{1})_{+}\right)$. Define $\Delta _{2}=\mathbb{E}_{\mu }(X^{(n)})\cdot 
\mathbb{E}_{\nu }\left(\left( 2-\frac{1}{X_{1}}\right) _{+}\right)-\mathbb{E}_{\mu
}\left(X^{(n)}\left( 2-\frac{1}{X_{1}}\right) _{+}\right).$ The random variable $\xi
=\max_{i=2...n}X_{i}$ is independent of $\left( 2-\frac{1}{X_{1}}\right)
_{+}$. Therefore,
\[
\Delta _{2}=\mathbb{E}_{\mu }(X^{(n)}-\xi )\cdot \mathbb{E}_{\nu }\left(
2-\frac{1}{X_{1}}\right) _{+}-\mathbb{E}_{\mu }\left((X^{(n)}-\xi )\left( 2-%
\frac{1}{X_{1}}\right) _{+} \right).
\]%
By definition, $X^{(n)}$ is greater than $\xi $. Hence $|\Delta
_{2}|\leq 2\mathbb{E}_{\mu }(X^{(n)}-\xi ).$ To estimate the difference
of expectations define $X_{-j}^{(n)}$ as $\max_{k=1,..n,\ j\neq i}X_{k}$.
Then $\mathbb{E}(X^{(n)}_{-j})=\mathbb{E}(\xi) $ for all $j$. If $%
X_{i}=X^{(n)}$, then all $X_{-j}^{(n)}$
except the one with $j=i$ coincide and are equal to $X^{(n)}$. Thus, $n%
\mathbb{E}(\xi) =\mathbb{E}\left(\sum_{j=1..n}X_{-j}^{(n)}\right)\geq (n-1)\mathbb{E}%
(X^{(n)})$ and $\mathbb{E}(X^{(n)})-\mathbb{E}(\xi) \leq \frac{\mathbb{E%
	}(X^{(n)})}{n}.$ Finally, $|\Delta _{2}|=O\left( \frac{\mathbb{E}_{\mu
	}(X^{(n)})}{n}\right) . $

\medskip
Let us estimate $\mathbb{E}_\mu \left(\left(X^{(n)}\right)^\alpha\right)$. For $\alpha>0$, we have $\mathbb{E}_\mu \left(\left(X^{(n)}\right)^\alpha\right)
=-\int_{0}^\infty t^\alpha d\,\mathbb{P}_\mu(\{X^{(n)}\geq t\}) $ and integration
by part gives 
\[
\alpha \int_{0}^\infty t^{\alpha-1}\mathbb{P}_\mu(\{X^{(n)}\geq t\}) dt
=\int_{0}^T+\int_{T}^\infty. 
\]

The first integral does not exceed $T^{\alpha }$. To estimate the second one
we combine the union bound with the Markov inequality: $\mathbb{P}_{\mu
}(\{X^{(n)}\geq t\})\leq n\mathbb{P}_{\nu }(\{X_{1}\geq t\})\leq n\frac{\mathbb{E}%
	_{\nu }((X_{1})^{\beta })}{t^{\beta }}.$ Therefore, 
\[
\alpha \int_{T}^{\infty }t^{\alpha -1}\mathbb{P}_{\mu }(\{X^{(n)}\geq t\})dt\leq \alpha n\mathbb{E}_{\nu }\left((X_{1})^{\beta }\right)\int_{T}^{\infty
}t^{\alpha -\beta -1}dt=\frac{\alpha }{\beta -\alpha }n\mathbb{E}_{\nu
}\left((X_{1})^{\beta }\right)\frac{1}{T^{\beta -\alpha }} 
\]%
for $\beta >\alpha $. Optimizing over $T$, we get $\mathbb{E}_{\mu }\left(\left(
X^{(n)}\right)^{\alpha }\right)\leq \left( \frac{\beta }{\beta -\alpha }\right)
\left( n\mathbb{E}_{\nu }\left((X_{1})^{\beta }\right)\right) ^{\frac{\alpha }{\beta }%
}=O\left( n^{\frac{\alpha }{\beta }}\right) .$

\medskip
It remains to put all the pieces together: 
\[
\Delta_0+\Delta_1+\Delta_2=O\left(\frac{1}{\sqrt{n}}\right)+O\left(\frac{ 
	\sqrt{E_\mu\left(|X^{(n)}|^{2}\right)} }{\sqrt{n}}\right)+O\left(\frac{\mathbb{E}_\mu
	(X^{(n)})}{n}\right)=O\left(\frac{1}{n^{\frac{1}{2}-\frac{1}{\beta}}}\right) 
\]
for any $\beta>2$ such that $\mathbb{E}_\nu (X_1)^\beta<\infty$. 
This
implies formula~(\ref{eq_asymp_good_iid}) for $\theta=1$.
\end{proof}

\medskip

\subsubsection{Proof of Lemma~\ref{lm2}} 
For unbounded distributions, $\mathbb{E}\left(X^{(n)}\right)$ tends to $+\infty$ and
thus by Proposition~\ref{prop6} the  ratio for $\varphi^1$ converges
to $\left(1-\mathbb{E}_\nu \left(2-\frac{1}{X_1}\right)_+\right)^{-1} $. Thus, the lower bound
immediately follows from the inequality $|x_1-1|\geq x_1- \left(2-\frac{1}{%
	x_1}\right)_+$.

For the upper bound, we have 
\[
\left(\pi(\varphi^1, \mathcal{P}_\infty(\nu))\right)^{-1}\geq \mathbb{E}_\nu \left(X_1- \left(2-\frac{1}{X_1}%
\right)_+\right)\geq \mathbb{E}_\nu\left(\left(X_1- \left(2-\frac{1}{X_1}%
\right)_+\right) 1_{\{X_1\geq 1\}}\right)= 
\]
\[
=\mathbb{E}_\nu\left(\left(X_1+\frac{1}{X_1}-2\right)1_{\{X_1\geq 1\}}\right)= \mathbb{E}%
_\nu\left(\left(\frac{(X_1-1)^2}{X_1}\right)1_{\{X_1\geq 1\}}\right)= \mathbb{E}_\nu
\left(g(X_1)1_{\{X_1\geq 1\}}\right), 
\]
where $1_A$ stands for the indicator of the event $A$. In order to relate
the expected value of $g(X_1)$ to $D$, we apply the Cauchy inequality 
\[
\frac{D}{2}=\mathbb{E}_\nu \left(|X_1-1|1_{\{X_1\geq 1\}}\right)=\mathbb{E}_\nu \left(\sqrt{%
	g(X_1)}1_{\{X_1\geq 1\}} \cdot \frac{|X_1-1|1_{\{X_1\geq 1\}}}{\sqrt{g(X_1)}}
\right)\leq 
\]
\[
\leq \sqrt{\mathbb{E}_\nu \left(g(X_1)1_{\{X_1\geq 1\}}\right)} \sqrt{\mathbb{E}_\nu\left(\frac{%
		(X_1-1)^2}{g(X_1)}1_{\{X_1\geq 1\}}\right)}. 
\]
The second factor on the right-hand side can be estimated as follows:
\[
\mathbb{E}_\nu\left(\frac{(X_1-1)^2}{g(X_1)}1_{\{X_1\geq 1\}}\right)=\mathbb{E}_\nu \left(X_1
1_{\{X_1\geq 1\}}\right)=\mathbb{E}_\nu \left(|X_1-1| 1_{\{X_1\geq 1\}}\right)+ \mathbb{E}_\nu
\left(1_{\{X_1\geq 1\}}\right)\leq \frac{D}{2}+1, 
\]
which completes the proof. \qed

\subsection{Bads}
\subsubsection{Not much weight around zero}
\begin{proposition}\label{prop7}
\textit{Consider a distribution }$%
\nu $\textit{ such that }$E_{\nu }(X_{1})=1$\textit{ and }$E_{\nu }\left(\frac{1}{%
	X_{1}}\right)<\infty $\textit{. Then the  ratio for the BH rule can
	be represented as} 
\begin{equation}
\pi(\varphi ^{1},\mathcal{P}_{n}(\nu ))=\frac{\mathbb{P}_{\nu }(\{X_{1}<T\})+\gamma \mathbb{P}_{\nu }(\{X_{1}=T\})}{\mathbb{E}_{\mu
	}\left(\min_{i\in N}X_{i}\right)}%
(1+o(1)),\ \ n\rightarrow \infty ,  \label{eq_asymp_bad_iid}
\end{equation}%
\textit{where }$T>0$\textit{ and }$\gamma,$ $0\leq \gamma <1,$\textit{ a%
	\textit{re} defined by the following condition:\footnote{%
		Formulas simplify for continuous distribution because $\mathbb{P}(X_{i}=T)=0$
		for all $T$ and thus we can always pick $\gamma =0$.}} 
\[
\mathbb{E}_{\nu }\left(\frac{1_{\{X_{1}<T\}}}{X_{1}}\right)+\gamma \mathbb{P}(\{X_{1}=T\})%
\frac{1}{T}=1. 
\]%
{For the Proportional rule,} 
\begin{equation}
\pi(\varphi ^{pro},\mathcal{P}_{n}(\nu ))=\frac{1}{\mathbb{E}_{\mu }(\min_{i\in N}X_{i})
	\cdot {\mathbb{E}_{\nu }\left(\frac{1}{X_{1}}\right)}}(1+o(1)).  \label{eq_asymp_bad_prop}
\end{equation}	
\end{proposition}	
\begin{proof} As in the proof of Proposition~\ref{prop6}, the symmetry of the problem implies $%
S(\varphi^1,\mathcal{P}_n(\nu))=n\mathbb{E}_{\mu} \left(X_1\varphi^1_1(X)\right)$ and
hence it is enough to estimate the contribution of one agent. We will
calculate this expectation in two steps: assuming first that $X_1=z$ is
fixed and averaging over $X_{j}$, $j\geq 2$, and then averaging over $z$.

Consider $\mathbb{E}_{\mu}\left(n X_1\varphi^1_1(X)\mid X_1=z\right)$. By
the definition of the BH rule we get 
\begin{equation}  \label{eq_big}
n\cdot X_1\varphi_1(X)\big\vert_{X_1=z}=\frac{X_{N\setminus 1}}{(n-1)}\cdot
1_{Q}+ z\cdot \frac{1-\sum_{j: X_j< z}\frac{1}{n}\frac{X_{N\setminus j}}{
		(n-1)X_j} }{|\{j\in N: \ X_j=z \}|/n}\cdot 1_{Q^{\prime }},
\end{equation}
where $Q$ is the event that $\sum_{j: X_j\leq z}\frac{X_{N\setminus j}}{
	n(n-1)X_j}\leq 1$ (in other words, $i$ belongs to the group of agents whose
share is given by the first line of equation~(\ref{18})) and the event $Q^{\prime }$ tells us
that the share of agent $1$ comes from the second line of~(\ref{18}), i.e., $\sum_{j:
	X_j< z}\frac{X_{N\setminus j}}{n(n-1)X_j}< 1< \sum_{j: X_j\leq z}\frac{
	X_{N\setminus j}}{n(n-1)X_j}$.

Let us apply the strong law of large numbers to~(\ref{eq_big}). Then, $\frac{
	X_{N\setminus 1}}{n-1}$ converges to $1$ almost surely, and the sum $\sum_{j:
	X_j\leq z}\frac{X_{N\setminus j}}{n(n-1)X_j}$ from the definition of $Q$
converges to $\mathbb{E}_\nu\left(\frac{1}{X_j}\cdot 1_{\{X_j\leq
	z\}}\right) $. Therefore, the first summand of~(\ref{eq_big}) tends to $%
1_{\{z<T\}}$, where $T $ is defined as $\inf \left\{T^{\prime }\,\big\vert\, \mathbb{E}%
_\nu\left(\frac{1_{\{X_j \leq T^{\prime }\}}}{X_j}\right)\geq 1 \right\}$. Thus, the asymptotic
contribution of the first term to $S(\varphi^1,\mathcal{P})$ is $\mathbb{P}%
_\nu(\{X_1<T\})$.

A similar application of the law of large numbers allows us to compute the
contribution of the second summand. We omit these computations.
\end{proof}

\subsubsection{Singularity at zero}
\begin{lemma}\label{lm5}
{ If a distribution $\nu$ has an atom at
	zero, then the BH and Proportional rules achieve { the optimal social cost}  in the limit:} 
\[
\pi(\varphi^{1},\mathcal{P}_{\infty}(\nu ))=\pi(\varphi^{pro},\mathcal{P}_{\infty}(\nu ))=1. 
\]

{If there is no atom and $\nu$ has a continuous density $f$ on $(0,a]$,
	but this density is unbounded, namely, $f(x)=\frac{\lambda}{x^\alpha}(1+o(1))$ as $%
	x\to +0$ for some $\lambda>0$ and $\alpha\in (0,1)$, then} 
\[
\pi(\varphi^{1},\mathcal{P}_{\infty}(\nu ))=1; \ \  \mbox{ however, } \ \ \pi(\varphi^{pro},\mathcal{P}_{n}(\nu )) =\Omega\left({n}\right). 
\]
\end{lemma}
\begin{proof}[Sketch of the proof.] In the case of an atom, there is an agent $i$ having $X_i=0$ with high
probability for large $n$. In such a situation, both rules $\varphi^1$ and $%
\varphi^{pro}$ coincide with the Utilitarian rule and therefore their
ratios are equal to $1$.

The second statement is proved similarly to Lemma~\ref{lm4}. For such $\nu$, the
expected value of the order statistic $X^{(k)}$ for small $k$ equals $\left(%
\frac{ 1-\alpha}{\lambda}\frac{k}{n}\right)^\frac{1}{1-\alpha}\cdot(1+o(1))$. Therefore, only the agent $i$ with $X_i=\min_{j}X_j$ receives a bad under
the BH rule with high probability, which gives $\pi(\varphi^{1},\mathcal{P}_{\infty}(\nu ))=1 $. The  argument for the Proportional rule is similar and therefore omitted.
\end{proof}

\end{document}